\documentclass[final,11pt,letterpaper]{scrartcl}  

 \usepackage[top=1in, bottom=2.8cm, left=1in, right=1in]{geometry}

\usepackage{etex}
\usepackage[utf8]{inputenc}
\usepackage{mathtools}
\usepackage{todonotes}
\usepackage[english]{babel}
\usepackage{amsmath,amssymb,tabularx,gastex,hhline,rotating,enumerate,xspace}
\usepackage{amsthm}
\usepackage[colorlinks,final]{hyperref}
\usepackage{tikz}
\usetikzlibrary{snakes}
\usepackage{pgfplots}

\usepackage[capitalise]{cleveref}

\newcommand{\comment}[1]{}

\newenvironment{PK}{\noindent\color{red} PK : }{}

\newenvironment{aw}{\noindent\color{magenta} AW :  }{}

\theoremstyle{plain}
\newtheorem{theorem}{Theorem}

\newtheorem{lemma}[theorem]{Lemma}

\newtheorem{corollary}[theorem]{Corollary}
\newtheorem{fact}[theorem]{Fact}
\newtheorem{conj}[theorem]{Conjecture}

\theoremstyle{definition}
\newtheorem{definition}[theorem]{Definition}
\newtheorem{remark}[theorem]{Remark}

\newcommand{\abs}[1]{\left|\mathinner{#1}\right|}

\newcommand{\floor}[1]{\left\lfloor\mathinner{#1} \right\rfloor}

\newcommand{\N}{\ensuremath{\mathbb{N}}}

\newcommand{\NP}{\ensuremath{\mathsf{NP}}\xspace} %

\newcommand{\ACC}{\ensuremath{\mathsf{ACC}^0}\xspace} %

\newcommand{\AC}{\ensuremath{\mathsf{AC}^0}\xspace}

\newcommand{\CC}{\ensuremath{\mathsf{CC}}\xspace}

\newcommand{\nudfa}{\ensuremath{\mathsf{NuDFA}}\xspace}

\renewcommand{\phi}{\varphi}
\newcommand{\eps}{\varepsilon}
\renewcommand{\epsilon}{\varepsilon}

\newcommand{\cH}{\mathcal{H}}
\newcommand{\cC}{\mathcal{C}}

\newcommand{\cQ}{\mathcal{Q}}

\newcommand{\cP}{\mathcal{P}}
\newcommand{\cS}{\mathcal{S}}

\newcommand{\smalloverline}[1]
{{\mspace{1mu}\overline{\mspace{-1mu}#1\mspace{-1mu}}\mspace{1mu}}}
\newcommand{\ov}[1]{\smalloverline{#1}}

\newcommand{\sse}{\subseteq}

\newcommand\ie{i.e., }
\newcommand\eg{e.g.\xspace}

\newcommand{\cdhcirc}[3]{\textsc{AND}_{#1} \circ \textsc{MOD}_{#2} \circ \textsc{MOD}_{#3}}
\newcommand{\modmod}[2]{\textsc{MOD}_{#1} \circ \textsc{MOD}_{#2}}
\newcommand{\cdhexpr}[3]{\textsc{AND}_{#1} \circ \textsc{MOD}_{#2} \circ \Sigma_{#3}}
\newcommand{\cdhandmod}[2]{\textsc{AND}_{#1} \circ \textsc{MOD}_{#2}}
\newcommand{\symp}[2]{\operatorname{sp}(#1, #2)}
\newcommand{\psymp}[2]{\operatorname{psp}(#1, #2)}
\newcommand{\setsum}[1]{\Sigma(#1)}

\newcommand{\aut}[1]{\operatorname{Aut}(#1)}
\newcommand{\sym}[1]{\operatorname{Sym}(#1)}

\newcommand{\F}{\mathbb{F}}
\newcommand{\GF}[1]{\ensuremath{\F_{#1}}}

\newcommand{\parmod}[1]{\;\, (\operatorname{mod} #1)}

\newcommand{\cAND}[1]{\textsc{AND}_{#1}}
\newcommand{\cMOD}[1]{\textsc{MOD}_{#1}}

\newcommand{\po}[1]{\mathbf{#1}}

\renewcommand{\setminus}{\mysetminus}
\newcommand{\mysetminusD}{\raisebox{.8pt}{\hbox{\tikz{\draw[line width=0.6pt,line cap=round] (3.5pt,0pt) -- (0,5.2pt);}}}}
\newcommand{\mysetminusT}{\mysetminusD}
\newcommand{\mysetminusS}{\raisebox{.5pt}{\hbox{\tikz{\draw[line width=0.45pt,line cap=round] (2.2pt,0) -- (0,3.8pt);}}}}
\newcommand{\mysetminusSS}{\raisebox{.35pt}{\hbox{\tikz{\draw[line width=0.4pt,line cap=round] (1.5pt,0) -- (0,2.8pt);}}}}

\newcommand{\mysetminus}{\mathbin{\mathchoice{\mysetminusD}{\mysetminusT}{\mysetminusS}{\mysetminusSS}}}

\begin{document}
 
\title{Violating Constant Degree Hypothesis\\ Requires Breaking Symmetry}
\author{Piotr Kawałek$^1$, Armin Weiß$^2$\\[1mm]
\small $^1$ Vienna University of Technology\\
\small$^2$ University of Stuttgart}


\maketitle

\begin{abstract}
The Constant Degree Hypothesis was introduced by Barrington et.\ al.\ \cite{BarringtonST90} to study some extensions of $q$-groups by nilpotent groups and the power of these groups in a certain computational model.  In its simplest formulation, it establishes exponential lower bounds for $\cdhcirc{d}{m}{q}$ circuits computing $\cAND{}$ of unbounded arity $n$ (for constant integers $d,m$ and a prime $q$). While it has been proved in some special cases (including $d=1$), it remains wide open in its general form for over 30 years. 

In this paper we prove that the hypothesis holds when we restrict our attention to symmetric circuits with $m$ being a prime.  
While we build upon techniques by Grolmusz and Tardos \cite{GrolmuszT00}, we have to prove a new symmetric version of their \emph{Degree Decreasing Lemma} and apply it in a highly non-trivial way. 
Moreover, to establish the result, we perform a careful analysis of automorphism groups of $\cdhandmod{d}{p}$ subcircuits and study the periodic behaviour of the computed functions. 

Finally, our methods also yield lower bounds when $d$ is treated as a function of $n$.
\\


\medskip
\noindent\textbf{Funding:}\\ 
\textbf{Piotr Kawałek:} Funded by the European Union (ERC, POCOCOP, 101071674). Views and opinions expressed are however those of the author(s) only and do not necessarily reflect those of the European Union or the European Research Council Executive Agency. Neither the European Union nor the granting authority can be held responsible for them. This research is partially supported by Polish NCN Grant \#2021/41/N/ST6/03907.\\
\textbf{Armin Weiß:} Partially funded by German DFG Grant WE 6835/1-2.

\end{abstract}
 \thispagestyle{empty}

\newpage

\pagestyle{plain}
\setcounter{page}{1}
\section*{Introduction}

Establishing strong lower bounds for general Boolean circuits represents one of the paramount and yet unattained objectives in the field of Computational Complexity Theory.
Whenever such lower bounds can be obtained, it is usually in some very restricted setting.
One of the standard limitations imposed on circuits in this context is the restriction of their depth.
Some strong results were obtained when the circuits have depth bounded by a constant $h$ and are built of unbounded fan-in Boolean $\textsc{AND}/\textsc{OR}$ gates and unary $\neg$ gates (so-called \AC circuits).
By a classical result of Furst, Saxe and Sipser \cite{fss84}, proved independently by Ajtai \cite{Ajtai83}, polynomial-size \AC circuits cannot compute the \textsc{PARITY} function (\ie the sum of the input bits modulo 2). In fact, a followup paper by Yao \cite{Yao85} strengthens the lower bound for $n$-ary \textsc{PARITY} to be of the form $2^{\Omega(n^c)}$, with a final result of H{\aa}stad \cite{Hastadphd} finding a precise $c = \frac{1}{h-1}$.

Here, a natural dual question arises: can modulo counting gates represent the $n$-ary Boolean $\cAND{n}$ function in the bounded-depth setting?
To be more precise, a $\CC_h[m]$ circuit is a circuit of depth $h$ using only (unbounded fan-in) $\cMOD{m}^{R}$ gates.
Each such gate sums the inputs modulo $m$ and outputs $1$ if the sum belongs to the set $R$ (we allow different $R \subseteq [m]$ for different gates), otherwise it outputs $0$.
Thus, the question is, after fixing $h$ and $m$, what size does a $\CC_h[m]$ circuit require to compute $\cAND{n}$?
Is there a polynomial-size construction for $\cAND{n}$, making the class $\ACC$ collapse to $\CC^0$  (where $\CC^0 = \bigcup_{h,m}\CC_h[m]$)?
The first question has a trivial answer when $m$ is a prime power, as then $\CC_h[m]$ circuits can express only bounded-arity \textsc{AND} (see \cite{BarringtonST90} or \cite{IdziakKK22LICS} for more details).
Surprisingly, for $m$ having multiple prime divisors, only slightly super-linear lower bounds are known \cite{chat-lowerbounds} and only for the number of wires
~-- even more: to the best of our knowledge it is consistent with the current understanding that $\NP \sse \CC_3[6]$. At the same time, the current best construction for $\cAND{n}$ has size $2^{O(n^c)}$ \cite{ChapmanW22,IdziakKK22LICS} for some constant $c$ depending on $h$ and $m$. 

This huge gap between lower and upper bounds suggests that the problem of establishing lower bounds in this context is very difficult.
Hence, one can consider simpler computational models before answering the above more general questions.
Interestingly, group theory outlines interesting in-between steps which can be considered in this context.
Barrington, Straubing and Thérien \cite{BarringtonST90} studied a model of non-uniform DFA (\nudfa) over groups (or, more generally, monoids), which they used to recognize Boolean languages.
They discovered that, if a group is an extension of a $p$-group by an abelian group, then its corresponding \textsc{\nudfa} can recognize all languages (however, most of them in exponential size).
Nevertheless, such \nudfa{}s cannot compute $\cAND{n}$ unless they have size at least $2^{\Omega(n)}$ \cite{BarringtonST90}.
Later this result was restated in a circuit language, saying that any $2$-level $\modmod{m}{q}$ circuit computing $\cAND{n}$ requires size $2^{\Omega(n)}$ \cite{GrolmuszT00, Grolmusz01, straubing2006note} (where $m$ is an integer and $q$ is a prime~-- here, other than in \cite{GrolmuszT00, Grolmusz01, straubing2006note}, the circuits have to be read that the $\cMOD{q}$ gate is the output gate).
The equivalence of the two statements is due to the fact that these (solvable) groups have internal structure based on modulo counting. 
The authors of \cite{BarringtonST90} conjectured that this $2^{\Omega(n)}$ lower bound generalizes to \nudfa{}s over extensions of nilpotent groups by $p$-groups.
This again can be reformulated to a $2^{\Omega(n)}$ lower bound for $\cdhcirc{d}{m}{q}$ circuits computing $\cAND{n}$ (see \cite{GrolmuszT00}).
This conjecture is known as Constant Degree Hypothesis (CDH for short), whose name corresponds to adding a layer of constant-arity $\textsc{AND}_d$ gates on the input level to a $\modmod{m}{q}$ circuit.
Interestingly enough, recently in \cite{IdziakKKW22-icalp} it was proven that all the other groups (which do no correspond to CDH) do not admit this lower bound, i.e.\ one can construct $\cAND{n}$ of size $2^{O(n^c)}$ for some $c<1$ using \nudfa{}s (or corresponding circuits) over these groups.
So $\cdhcirc{d}{m}{q}$ circuits are really the only (algebraically) natural subclass of $\CC^0$ circuits for which these strong $2^{\Omega(n)}$ lower bounds can hold.

Low-level $\CC^0$ circuits have many surprising connections.
For instance, the techniques used in the construction of relatively small $\CC^0$ circuits for the $\textsc{OR}_n$ function (equivalently, $\cAND{n}$) found in \cite{BarringtonBR94} are useful in constructing small explicit Ramsey-type graphs \cite{grolmusz2000, gopalan2014}.
These constructions are also used to produce better locally-decodable error-correcting codes \cite{efre12, zev11}, private information retrieval schemes \cite{zeev15}, and secret sharing schemes \cite{liu2018}.
The lower bounds for codes considered in \cite{efre12} imply lower bounds for certain $\CC^0$ circuits.
On the contrary, good lower bounds for low-level $\CC^0$ circuits imply faster algorithms for solving equations in solvable groups \cite{IdziakKKW22-icalp}, faster algorithms for certain algebraic versions of circuit satisfiability problems \cite{IdziakKK20} and also faster algorithms for some variants of the Constraint Satisfaction Problem with Global Constraints \cite{brakensiek2022}. 

These diverse interconnections encourage to put even more effort to find the correct sizes for optimal modulo-counting circuits computing $\cAND{n}$.
In this pursue, proving (or disproving) $\textsc{CDH}$ plays a central role. The hypothesis is already proven in several special cases: in particular, the case $d=1$  was confirmed in the very same paper the hypothesis was defined.
Moreover, if there is a bound on the number of $\cAND{d}$ gates that are wired to each $\cMOD{m}$ gate, the desired lower bound is also true \cite{GrolmuszT00}. More precisely, the number of such connections is required to be $o(\frac{n^2}{\log n})$.
The technique used in this case is based on the so-called \textit{Degree Decreasing Lemma}, whose name corresponds to gradually decreasing the degree $d$, which eventually leads to the $d=1$ case.
The Degree Decreasing Lemma can also be used when the polynomials over $\mathbb{Z}_m$ corresponding to the $\cdhandmod{d}{m}$ part of the circuit can be written using a sublinear number of binary multiplications \cite{grolmusz2000}. 

In many studies of different circuit complexity classes, \textit{symmetry} seems to play an important role. In this context both symmetric circuits, as well as symmetric functions were considered. Here, symmetry for a circuit/function means that permuting its inputs/variables does not change the considered circuit/function. 
Let us here mention the recent results on lower bounds for symmetric arithmetic circuits for the permanent and also a construction of short symmetric circuits for the determinant \cite{DawarW20} as well as the lower bound from \cite{HeR23} for computing a certain entry in a product of matrices (here, symmetry means invariance under permuting rows and columns of matrices). 

\textit{Symmetry} seems to play also a special role for $\CC_h[m]$ circuits. The remarkable construction of relatively small circuits for $\cAND{n}$ in \cite{BarringtonBR94} uses symmetric polynomials as an intermediate object before translating them to circuits. This translation, when done carefully, leads also to symmetric circuits. Similarly, some of the newer, more optimal constructions of two level $\CC_2[m]$ circuits for $\cAND{n}$ can be performed fully symmetrically \cite{ChapmanW22, IdziakKK22LICS}.
Additionally, \cite{GrolmuszT00, Grolmusz01, straubing2006note} analyze the periodic behaviour of the symmetric functions that can be represented by small (not necessarily symmetric) $\modmod{m}{q}$ circuits. A value of a symmetric Boolean function $f(x_1, \ldots, x_n)$ is determined by the number of ones among $x_1, \ldots, x_n$.
Hence, for an integer $0\leq m \leq n$, we can naturally define $f(m)$ as $f(1^m0^{n-m})$ and say that an integer $r$ is a period of $f$ whenever $f(m+r) = f(m)$ for all $0\leq m \leq n-r$.
It follows from \cite{GrolmuszT00, straubing2006note} that the only symmetric functions that have representations as $\modmod{m}{q}$ circuits of subexponential size must have periods of the form $m\cdot q^k$ with $m\cdot q^k \leq n$.
In particular, $\cAND{n}$ must have exponential-size circuits. 

In this paper we prove that symmetric $\cdhcirc{d}{p}{q}$ circuits computing $\cAND{n}$ have exponential size.  
A key to the proof is to analyze the periodic behaviour of the functions computed by such circuits.
 Our techniques work also when $d$ is unbounded and is considered as a function of $n$. The following theorem characterizes the periodic behaviour such functions. 

\begin{theorem}\label{thm:mainPeriod}
   Let $p$ and $q$ be primes and $n\geq 13$ and let $1 \leq d \leq n$. Then any function computed by an $n$-input symmetric $\cdhcirc{d}{p}{q}$ circuit of size $s < 2^{n/9}$ has period $p^{k_p}q^{k_q}$ given that $p^{k_p} > d$ and $q^{k_q} > \log s + 1$.
\end{theorem}

This theorem can be read in comparison to \cite{DenenbergGS86}, where polynomial-size symmetric $\AC$ circuits of arity $n$ are shown to represent only functions that are constant on the interval $\{n^\eps, \dots, n-n^{\eps}\}$ (for large enough $n$). The next theorem follows from a careful application of \cref{thm:mainPeriod}.

\begin{theorem}\label{thm:mainGeneralIntro}
   Let $p$ and $q$ be primes, let $n$ be a large enough integer, and let $d\leq \sqrt{n}$. Then every symmetric $\cdhcirc{d}{p}{q}$ circuit computing the $\cAND{n}$ function has size at least $2^{ n/(2dpq)}$.
\end{theorem}

Note that the restriction $d\leq \sqrt{n}$ still includes the most interesting case.
Indeed, for $\sqrt{n}\leq d \leq n-\sqrt{n}$ we get an almost trivial lower bound of $2^{\sqrt{n}}$ (see \cref{thm:mainGeneralProofs}).
Moreover, \cref{thm:mainGeneralIntro} suggests an interesting trade-of between the degree and the size at $d\approx \sqrt{n}$.
Then we can reach a lower bound for the size of the form $2^{\Omega(\sqrt{n})}$.

As a direct consequence of \cref{thm:mainGeneralIntro} we get the desired result for $\cAND{n}$.

\begin{corollary}\label{thm:mainCDH}
   For constant $d$, and primes $p$, and $q$ every symmetric $\cdhcirc{d}{p}{q}$ circuit for $\cAND{n}$ has size at least $2^{\Omega( n)}$. Thus, CDH holds for symmetric circuits.
\end{corollary}

Before we go to the more technical part, let us briefly mention an opposite perspective on the results of this paper. Although current evidence seems to support $\textsc{CDH}$ and lower bounds for $\cAND{n}$ for general $\CC_h[m]$ circuits, it is known that $\CC_h[m]$ circuits using $O(\log n)$ random bits are able to compute $\cAND{n}$ in polynomial size \cite{HansenK10}.
This was even improved in \cite{IdziakK22}, by showing that $\modmod{p}{q}$ circuits can also be used for representing $\cAND{n}$ in this probabilistic model. This might be interpreted as a strong argument against lower bounds, because now it is enough to derandomize the construction for $\modmod{p}{q}$ circuits. We already understand these $2$-level circuits relatively well (to the point that we can prove strong lower bounds for them for the $\cAND{n}$ function itself). Our \cref{thm:mainCDH} implies one cannot construct  $\textsc{AND}_n$ with polynomial-size symmetric $\cdhcirc{d}{p}{q}$ circuits. Hence, to make short deterministic constructions one needs to either go beyond the symmetric setting or consider larger depths.

\paragraph*{Outline.}

The paper is organized as follows: in \cref{prelims} we fix our notation on circuits as well as hypergraphs and group actions. These notions are essential in the later study of the symmetric structure of circuits. In \cref{preparation}, we describe how to rewrite a circuit into a nicer form that we use throughout the paper. Then in \cref{main-sec} we present our key lemmas and show how to derive our main results from them. The proofs of these lemmas are deferred to \cref{proofs-sec}. In \cref{further-remarks} we add some discussion of our results.

\section{Preliminaries}\label{prelims}

For $d \in \N$ we write $[d]$ for the set of integers $\{1, \dots, d\}$.

\paragraph*{Hypergraphs.}
For a set $X$ we denote its power set by $\cP(X)$. 
A \emph{hypergraph} on a set of vertices $V$ is a pair $(V,
E)$ with $E \sse \cP(V)\setminus\{\emptyset\}$.
\newcommand{\hyp}[2]{\cH_{#1}^{#2}}
A \emph{$\F_p$-labeled hypergraph} is a pair $G=(V, \lambda)$ where $\lambda \colon (\cP(V) \setminus\{\emptyset\} )\to \F_p$. We obtain an (unlabeled) hypergraph by setting $E  = \{ e\sse V \mid \lambda(e) \neq 0\}$ and call each $e$ with $\lambda(e)\neq 0$ an \emph{edge}  of $G$. Thus, an $\F_p$-labeled hypergraph is indeed a hypergraph where we assign to each edge a number from $\F_p \setminus\{0\}$.
Moreover, $(V, \lambda)$ is called an \emph{$\F_p$-labeled $d$-hypergraph} if for all $e \in \cP(V)$ with $\abs{e} > d$ we have $\lambda(e) = 0$. We write $\hyp{p}{d}(V)$ for the set of $\F_p$-labeled $d$-hypergraphs on $V$.  For $C\sse V$ we write $\ov C = V \setminus C$ for the complement of $C$.

If $G=(V,\lambda)$ and $H = (V,\zeta)$ are $\F_p$-labeled hypergraphs on the same set of vertices $V$, we define $G + H$ (resp.\ $G-H$) as $(V, \lambda + \zeta)$ (resp.\ $(V, \lambda - \zeta)$) where $\lambda + \zeta$ denotes the point-wise addition.
We interpret any subset $E \sse \cP(V)\setminus\{\emptyset\}$ as a hypergraph by setting $\lambda(e) = 1$ if $e \in E$ and $\lambda(e) = 0$
 otherwise (be aware of the slight ambiguity as $V$ is not uniquely defined by $E$~-- but it always will be clear from the context). Thus, we have defined the addition $G + E$ (resp.\ $G - E$). We extend this to $G + e = G + \{e\}$. On the other hand, for $S \sse V$, we define $G + \Sigma(S)$ as $G + \{\{s\} \mid s \in S\}$.

 \paragraph*{Permutation groups.}

For any set $V$ we denote the group of permutations on $V$ by $\sym{V}$ (\ie the symmetric group). For an integer $n$ we write $\sym{n}$ or $S_n$ for the abstract symmetric group acting on any $n$-element set. 
Any subgroup $\Gamma \leq \sym{V}$ \emph{acts} faithfully on $V$ and is called a permutation group.
A subset $U \sse V$ is called an \emph{orbit} of the action of $\Gamma$ on $V$ if $U=G\cdot x$ for some $x \in V$. If there is only one orbit, the action of $\Gamma$ on $V$ is called transitive. Clearly, the orbits form a partition of $V$; moreover, if $U_1, \dots, U_k\sse V$ are the orbits of the action of $\Gamma \leq \sym{V}$ on $V$, then $\Gamma \leq \sym{U_1} \times \cdots \times \sym{U_k}$ where $\times$ denotes the direct product of groups.

Finally, let $\Gamma'\leq \Gamma$ be a subgroup. A \emph{left-transversal} (in the following simply \emph{transversal}) of $\Gamma'$ in $\Gamma$ is a subset $R \sse \Gamma$ such that $R$ is a system of representatives of $\Gamma/\Gamma'$~-- in other words, if $R \Gamma' = \Gamma$ and $r\Gamma' \cap s\Gamma'=\emptyset$ for $r,s \in R$ with $r\neq s$.
For further details on permutation groups, we refer to \cite{cameron99}.

 \paragraph*{Actions on hypergraphs.}

Given an action of $\sym{V}$ on $V$, it induces an action on $\cP(V)$. Moreover, this extends to an action on $\hyp{p}{d}(V)$ where a permutation $\pi \in \sym{V}$ maps $(V,\lambda)$ to $\pi((V, \lambda)) = (V, {}^\pi\!\lambda)$ for
${}^\pi\!\lambda$ defined by $({}^\pi\!\lambda)(e) = \lambda(\pi^{-1}(e))$. Be aware that the ${}^{-1}$ is not by accident but rather guarantees that if some $e \in \cP(V)$ has label $\gamma = \lambda(e)$, then $\pi(e)$ has label ${}^\pi\!\lambda(\pi(e)) = \lambda(e)$. Note that two labeled hypergraphs with vertices $V$ are isomorphic if and only if they are in the same orbit under $\sym{V}$. 
 A permutation $\pi \in \sym{V}$ is called an \emph{automorphism} of $G = (V,\lambda)$ if $\pi(G) = G$~-- with other words, if $\lambda(\pi(e)) = \lambda(e)$ for all $e \in \cP(V)$. 
For a labeled hypergraph $G$, we denote its group of automorphisms by $\aut{G}$.

\paragraph*{Circuits.}
A circuit is usually defined as a directed acyclic graph with labels on its vertices that inform what kind of operation (like for instance $\land, \lor, \neg$, $\cMOD{p}^R$) a given vertex (gate) computes.
A depth-$d$ circuit of arity $n$ is a circuit that consists of $n$ inputs gates $x_1, \ldots, x_n$ and $d$ layers (or levels) $G_1, \ldots, G_d$ of inner gates (we do not count the input gate as a level). Between neighbour layers $G_{i-1}$ and $G_{i}$ there is a layer of wires $W_i$ which contains directed edges between $g \in G_{i-1}$ and $h \in G_{i}$ (where $G_0 = \{x_0, \dots, x_n\}$). We allow for multiple (directed) edges between the same pair of gates. Moreover, gates are labeled with necessary information which allows to compute a function they represent. In our case we use $\cMOD{p}^R$ gates where $p$ is a prime and $R \sse \F_p$. A $\cMOD{p}^R$ with inputs $y_1, \dots, y_k$ outputs $1$ if and only if the sum of its inputs modulo $p$ is contained in $R$. A circuit is called an \emph{expression} if it is a tree when removing the input layer. A subexpression of an expression is a subgraph containing for every gate also all its predecessors (towards the input gates).
For circuits $C,D$ with $n$ inputs we write $C \equiv D$ if for all inputs $\ov b \in \{0,1\}^n$ they evaluate to the same value. 
We define the \emph{size} of a circuit as its number of non-input gates.

In this article we consider $\cdhcirc{d}{p}{q}$ circuits: such a circuit consist of $3$-levels. On level $1$ there are $\cAND{d}$ gates each of which receives inputs from at most $d$ input gates. The second level $G_2$ consists of $\cMOD{p}^R$ gates~-- each of them is labeled with an accepting set $R\subseteq \{0, \ldots, p-1\}$.  The output layer $G_3$ contains only one $\cMOD{q}^R$ gate, which sums all the wires from $W_3$ modulo $q$.

We say that a circuit $C$ is symmetric if no permutation of the input wires changes the circuit.
Note that here the word \textit{symmetric} refers to to a \textit{syntactic} structure of a circuit, rather than a semantic property of the computed function computed by it.
More formally, a circuit $C$ on inputs $x_1, \dots, x_n$ is called \emph{symmetric} if for any $\pi\in \sym{\{x_1, \dots, x_n\}}$ there is a permutation $\pi'$ on the set of gates extending $\pi$ (meaning that $\pi(x_i) = \pi'(x_i)$ for all $i\in [n]$) such that there are $k$ wires connecting gate $i$ to gate $j$ if and only if there are $k$ wires connecting gates $\pi'(j)$ to gate $\pi'(j)$.

\section{Preparation: Circuits, Expressions and Hypergraphs}\label{preparation}

For a simpler notation of expressions, let us denote $\cMOD{p}^R$ with inputs $y_1, \dots, y_k$ instead by $\po b(\sum_{i=1}^k y_i ; R)$ for $R \sse \F_p$, where $\po b$ computes the function
 \[\po b (y; R) = \begin{cases}
    1& \text{if } y \in R\\
    0& \text{if } y \not\in R.
\end{cases}\]
Be aware that we use $\po b$ for different domains, i.e.\ as a function $\F_p \to \{0,1\}$ and $\F_q \to \{0,1\}$. The domain will be clear from the context.

\paragraph*{From circuits to expressions.}

Any 2-level $\cdhandmod{d}{p}$ circuit computes a polynomial function over the field $\GF{p}$. Indeed, the $\cAND{d}$ gates act like a multiplications on the two element domain $\{0,1\}\subseteq \GF{p}$ and the $\textsc{MOD}_p^R$ gate sums the results and checks whether the sum belong to the accepting set $R$.
Because our circuits are of constant-depth, we can unfold the circuits to obtain expressions. Note that this might lead to a polynomial blow-up in size (more precisely, a circuit with size $s$  and  depth bounded by $h$ is converted to an expression of size at most $s^{h-1}$). 
Moreover, note that unfolding the circuit does not destroy the property of being symmetric.
Hence, every symmetric $\cdhcirc{d}{p}{q}$ circuit yields also a symmetric expression
\vspace{-1mm}
\begin{equation}\label{eqn:cdhexpr}
\po b \bigl(\sum_{i=0} ^l \alpha_i\, \po b(\po p_i(\ov x) ; R_i); R\bigr)
\end{equation}
for suitable $\alpha_i \in \F_q$, $ R_i \sse \F_p$ and polynomials $\po p_i$ of degree bounded by $d$ for $i\in \{1, \dots, l\}$ and $R \sse \F_q$ which computes the same function. Here, $l$ is the number of $\cMOD{p}$ gates used in the $\cdhcirc{d}{p}{q}$ circuit, while $\alpha_i$ tells us how many times a given $\cMOD{p}$ gate is wired to the $\cMOD{q}$ gate.
Let us take a closer look at what being symmetric means for an expression of the form \eqref{eqn:cdhexpr}:
for each $\pi \in \sym{n}$ there exist $\pi' \in S_l$ such that for all $i \in \{1,\ldots, l\}$ we have $\alpha_i = \alpha_{\pi'(i)}$, $R_i = R_{\pi'(i)}$, and $\po p_i(x_1, \ldots, x_n) = \po p_{\pi'(i)}(x_{\pi(1)}, \ldots, x_{\pi (n)})$ (here = refers to equality in the polynomial ring $ \F_p[x_1, \dots, x_n]$).

As we are going to analyze the periodic behaviour of considered symmetric circuits, and the outer $\po b$ does not modify the period in any substantial way we will now concentrate on the symmetric expressions of the form
\vspace{-1mm}
\begin{equation}\label{eqn:inexpr}
\po f = \sum_{i=0} ^l \alpha_i\, \po b(\po p_i(\ov x) ; R_i)
\end{equation}
with $r_i \in  \F_p$, $\alpha_i \in \F_q$ and polynomials $\po p_i$ of degree bounded by $d$. We call such expressions $\cdhexpr{d}{p}{q}$ expressions and each $\po b(\po p_i(\ov x) ; r_i)$ an \emph{elementary subexpression} of $\po f$.
\newcommand{\pol}[2]{\operatorname{pol}(#1,#2)}

In the following, let us write $\po b(\po p(\ov x) ; r) $ for $\po b(\po p(\ov x) ; \{r\})$. Using this notation we have
$\po b(\po p(\ov x) ; R) = \sum_{r\in R} \po b(\po p(\ov x) ; r).$
Moreover, we always assume that for $i\neq j$ we have $(\po p_i, r_i) \neq (\po p_j, r_j)$ as otherwise we can replace $\alpha_i\, \po b(\po p_i(\ov x) ; r_i) + \alpha_j\, \po b(\po p_i(\ov x) ; r_j)$ by $\alpha_{ij}\, \po b(\po p_i(\ov x) ; r_i)$ where $\alpha_{ij} = \alpha_i + \alpha_j$. Thus, using $\pol{n}{d}$ to denote the set of multilinear polynomials  in $ \F_p[x_1, \dots, x_n]$ with degree bounded by $d$, we  rewrite $\po f $ in \eqref{eqn:inexpr} as
\vspace{-2mm}
\begin{equation}\label{eqn:normalizedexpr}
\po f = \sum_{\po p \in \pol{n}{d}}\sum_{r \in \F_p}\alpha_{\po p , r}\, \po b(\po p(\ov x) ; r).
\end{equation}
Note that now for size of $\po f$ we only need to count the non-zero $\alpha_{\po p,r}$ (plus the number of $\cAND{}$ gates computing the polynomials $\po p$).

\paragraph*{Polynomials and hypergraphs.}

Let us take a closer look at the $\cdhandmod{d}{p}$ part of a $\cdhexpr{d}{p}{q}$ circuit or expression.
As any such expression is represented by a polynomial of degree $d$, we will need to deal with these polynomials and their symmetries. Notice that without loss of generality, we can assume that the polynomial corresponding to a $\cdhandmod{d}{p}$ is multi-linear since, because values of variables are restricted to $\{0,1\}$ each occurrence of a higher power $x^k$ of a variable $x$ can be simply replaced by $x$.

In order to deal better with the combinatorics and symmetries of polynomials, we think of polynomials as hypergraphs.
A multilinear polynomial $\po p\in \F_p[x_1, \dots, x_n]$ with the degree bounded by $d$ can be naturally identified with an $\GF{p}$-labeled $d$-hypergraph $G = (V,\lambda)$ as follows:
\begin{enumerate}
    \item Treat each variable $x_i$ in $\po p(x_1, \ldots, x_n)$ as a vertex in the graph $G_{\po p}$. Thus, $V = \{x_1, \dots ,x_n\}$, which we also identify with the set $[n]$.
    \item Each monomial $\gamma \cdot x_1 \cdot \ldots \cdot x_{d}$ is represented by a hyperedge with a label $\gamma$, \ie we have $\lambda(\{x_1, \ldots, x_d\}) = \gamma$. 
\end{enumerate}

Thus, we get a one-to-one correspondence between multilinear polynomials over $\GF{p}$ of degree at most $d$ and $\GF{p}$-labeled $d$-hypergraphs. This means that we can also do the reverse -- for each labeled graph $ G$ we can create its corresponding polynomial $\po p_{G}$. Moreover, note that also the arithmetic operations we defined on hypergraphs as well as the group actions agree with those on polynomials. Therefore, in the following, we use polynomials and hypergraphs interchangeably.

Now, we can use our graph notation for polynomials in a more general setting and denote each expression $\po b(\po p(\ov x) ; r)$ by $\po b(G_{\po p}; r)$ or simply $\po b(G;r)$  (when we start with a hypergraph representing a given polynomial). Thus, we can reformulate any expression of the form \eqref{eqn:normalizedexpr} as $\sum_{\po p \in \pol{n}{d}} \sum_{r\in \F_p} \alpha_{\po p , r}\,\po b(G_{\po p}; r) = \sum_{G \in \hyp{p}{d}(V)} \sum_{r\in \F_p} \alpha_{G , r}\,\po b(G; r).$

\paragraph*{Symmetric expressions induced by hypergraphs.} 
Now we define several notions, useful in analysing symmetric expressions. 
For $G = (V,\lambda)$ and $\pi \in \sym{V}$, let us write $\po b^{\pi}(G;R)$ for $\po b(\pi G; R)$.
The action of $\sym{V}$ on $V$ now extends naturally to an action on expressions of the form $\po f = \sum_{G \in \hyp{p}{d}(V)} \sum_{r\in \F_p}\alpha_{G , r}\, \po b(G; r)$ by setting
\[\pi(\po f ) = \sum_{G \in \hyp{p}{d}(V)} \sum_{r\in \F_p} \alpha_{G , r}\, \po b^\pi(G ; r). \]
Now,  $\po f$ being symmetric can be simply expressed as the fact that for each  $\pi \in \sym{V}$ we have
$\pi(\po f ) = \po f.$

\begin{definition}
Let $G = (V,\lambda)$ be a labeled $d$-hypergraph. Let $\aut{G}$ be its group of automorphisms and let $\pi_1, \ldots, \pi_k$ be a transversal of $\sym{V}/\aut{G}$. For a given $r\in \GF{p}$, define $\po s(G; r)$ to be the following $\cdhexpr{d}{p}{q}$ expression
\vspace{-2mm}\begin{equation}\label{decomp-sum}
\po s(G; r) = \sum_{i=0}^k \po b^{\pi_i}(G, r).
\end{equation}
\end{definition}

\noindent One needs to check that the above definition does not depend on the choice of the transversal, as there is a choice in picking the specific traversal $\pi_1, \ldots, \pi_k$ which we use to create $\po s(G; r)$. However, as $G$ is invariant under its automorphisms, no matter how we choose the specific $\pi_1, \ldots, \pi_k$, we get the same expression in the end. In fact, every symmetric $\cdhexpr{d}{p}{q}$ expression containing $\po b(G; r)$ as subexpression, must also contain $\po s(G; r)$ as subexpression. So $\po s(G; r)$ is a symmetric closure of the basic expression $\po b(G;r)$. Let us summarize this as follows:
\begin{remark}\label{rem:symm-closure}
 For every labeled $d$-hypergraph $G$ and every $r\in \GF{p}$, the expression $\po s(G; r)$ is symmetric. Moreover, it is the smallest symmetric expression that contains $\po b(G; r)$ as an elementary subexpression.
\end{remark}


\begin{fact}\label{sym-sum}
Every symmetric $\cdhexpr{d}{p}{q}$ expression $\po f$ can be written as a sum 
\begin{equation*}
\po f(\ov x)= \sum_{ G \in \hyp{p}{d}(V)} \sum_{r \in \F_p} \beta_{G, r} \cdot  \po s(G; r).
\end{equation*}
for  $\beta_{G, r}\in \F_q$ (recall that $\hyp{p}{d}(V)$ denotes the set of labeled $d$-hypergraphs on $V$).
\end{fact}
 
\begin{proof}
If a symmetric $\po f$ has some $\beta \cdot \po b(G, r)$ as an elementary subexpression, it must also have $\beta \cdot \po s(G; r)$ as a subexpression (see \cref{rem:symm-closure}). But now $\po f - \beta \cdot \po s(G; r)$ is a symmetric expression which is shorter than $\po f$, and hence we can use induction to prove the desired decomposition for $\po f(\ov x)$, by adding $\beta \cdot \po s(G; r)$ to the decomposition of $\po f - \beta \cdot \po s(G; r)$.
\end{proof}

\section{High-level Description of the Proof}\label{main-sec}

We now start with an expression as in \cref{sym-sum} and prove our main theorems. For this, we need several definitions and intermediate results. The proofs of these intermediate results are deferred to \cref{proofs-sec}; instead we give some high-level ideas how the respective results are used and then show how our main results follow from the intermediate results.
As every symmetric expression $\po f$ is decomposed into an appropriate sum of elements of the form $\alpha \cdot  \po s(G;r)$, we need a deeper understanding of each $\po s(G;r)$.
We investigate these expressions $\po s(G;r)$ in three main steps:
\begin{enumerate}
    \item we analyze the symmetries of $G$ to find a large so-called \emph{fully symmetric set} (see \cref{automo-lem}),
    \item we process the hypergraph $G$ further to make it \emph{symmetry purified} (see \cref{symmetry-purified-def} and \cref{sym-expr}) applying two versions of the Degree Decreasing Lemma (\cref{ddl} and \cref{sddl}),
    \item we analyze the periods of the resulting expressions $\po s(G;r)$ (see \cref{period-sym-pure}).
\end{enumerate}

 \paragraph{Symmetries of hypergraphs.} Recall that one of our goals is to prove exponential lower bounds on the size of symmetric circuits/expressions computing $\cAND{n}$. We start with the observation that, if in an expression $\po f$ we find a very asymmetric graph $G$, we  know that the size of $\po f$ must be relatively large. This is because the automorphism group of $G$ is small and, hence, the length of the expression of the form \eqref{decomp-sum} induced by $G$, i.e.\ $\po s(G; R)$, must be large (more precisely, $k$ as defined above is large). On the other hand, for highly symmetric graphs $G$, we can find a big, very regular substructure of $G$, which we will call a \textit{pseudo-clique}.

\begin{definition}
Let $G$ be an $\GF{p}$-labeled hypergraph $G = (V,\lambda)$ (i.e. $\lambda: \cP(V) \setminus \{\emptyset\}\to \GF{p}$). We say that a subset $C \sse V$ is \emph{fully symmetric}, if for each pair of subsets $e_1, e_2\subseteq V$ with $|e_1| = |e_2|$  and $e_1 \cap \overline{C} = e_2 \cap \overline{C}$ we have $\lambda(e_1) = \lambda(e_2)$.

Moreover, a $\GF{p}$-labeled hypergraph $G=(V,\lambda)$ is called a \emph{pseudo-clique} if $\aut{G} = \sym{V}$~-- or, equivalently, if for each $d \in [n]$ there is some $\lambda_d$ such that $\lambda(e) = \lambda_d$ all $e\sse V$ with $\abs{e} = d$.
\end{definition}

Note that an induced subgraph on a fully symmetric subset of vertices is a pseudo-clique.
Now, we are ready to present a key lemma, which allows us to restrict our attention only to very symmetric hypergraphs.

\begin{lemma}\label{automo-lem}
Let $0<\eps < 1/8$. Every $\GF{p}$-labeled hypergraph $G = (V,\lambda)$ with $n =\abs{V} \geq 13$ either
\begin{itemize}
    \item has a fully symmetric subset on at least $n - \floor{\eps n}$ vertices, or
    \item its automorphism group satisfies $|\sym{V}/\aut{G})| > 2^{\floor{\eps n}}$. 
\end{itemize}
\end{lemma}

\paragraph*{Reduction based on the Degree Decreasing Lemma.}
One of the few examples of lower bounds for circuits using modulo counting are due to Grolmusz and Tardos \cite{GrolmuszT00,Grolmusz01}. The authors prove lower bounds for $\cdhcirc{d}{p}{q}$ circuits with restrictions put on connections between $\cAND{d}$ layer and $\textsc{MOD}_p$ gates. More precisely, \cite{Grolmusz01} shows that if the number of multiplications needed to compute the polynomial corresponding to each $\cdhandmod{d}{p}$ subcircuit is bounded by $cn$ for small enough $c$, then a $\cdhcirc{d}{p}{q}$ circuit requires exponential size to compute $\cAND{n}$. One of their key tools is the so-called Degree Decreasing Lemma:

\begin{lemma}[Degree Decreasing Lemma]\label{ddl}
Let $p \neq q$ be prime numbers.  Then every function $f: \GF{p}^3 \mapsto \GF{q}$  represented by a $3$-ary $\cdhexpr{2}{p}{q}$ expression
$\po b(\gamma \cdot z_1 \cdot z_2 + y; t)$
can be also represented by an expression of the form
$$\sum_{(j_1,j_2,j_3) \in \GF{p}^3} \sum_{r \in \GF{p}}  \beta_{j_1,j_2,j_3}^{(r)} \cdot \po b(j_1 z_1 + j_2 z_2 + j_3 y;\, r).$$
where $\beta_{j_1,j_2,j_3}^{(r)}$ are some coefficients from $\GF{q}$ (also depending on $\gamma$ and $t$).
\end{lemma}

The Lemma is a consequence of the result by Grolmusz \cite[Lemma 6]{Grolmusz01} (note that the statement there does not include the factor $\gamma$, so formally, to obtain \cref{ddl}, one needs to apply \cite[Lemma 6]{Grolmusz01} several times). One can also see it as a consequence of \cite[Fact 3.3]{IdziakKK22LICS}. This very simple lemma allows us to navigate through the space of different representations for a given function $f$ by allowing a local change of its corresponding expression.  The power of the lemma comes from the fact that we can substitute arbitrary polynomials for $z_1, z_2, y$ and obtain many different equivalences.

We will need a more regular version of the Lemma when the multiplication inside $\po b$ has bigger arity. The price  we pay for a nicer form is that the represented function has a smaller (partially Boolean) domain, which slightly reduces the scope of applicability of the lemma (as we cannot substitute any polynomial for the variables); however, it still suffices for our purposes.

\begin{lemma}[Symmetric Degree Decreasing Lemma]\label{sddl}
Let $p\neq q$ be prime. Let $\gamma \in \GF{p}\setminus \{0\}$. Then every function $f:\{0,1\}^d \times \GF{p} \mapsto \GF{q}$ represented by a $d+1$-ary $\cdhexpr{d}{p}{q}$ expression
\[\po b(\gamma \cdot x_1 \cdot \ldots \cdot x_d + y; t)\]
can be also represented by expression
 \[\po h(\ov x, y;t) = \po b(y; t) + \sum_{r \in \GF{p}} \beta_{t,r} \sum_{S \subseteq [d]} \alpha_{|S|} \cdot \po  b(\gamma\cdot \Sigma(S)  + y; r)\]
for $\alpha_{|S|} = (-1)^{|S|}$ and some coefficients $\beta_{t,r} \in \F_q$.
\end{lemma}

The key property of the formula $\po h$ is that it is invariant under permutations of the variables $x_1, \ldots, x_d$, which is not the case for original Degree Decreasing Lemma of \cite{Grolmusz01}.
The next step in the proof is to apply the  Symmetric Degree Decreasing Lemma (\cref{sddl}) to expressions generated by highly-symmetric hypergraphs in order to obtain an even nicer representation defined as follows:

\begin{definition}\label{symmetry-purified-def}
We call an $\GF{p}$-labeled $d$-hypergraph $G = (V,\lambda)$ \textit{symmetry-purified} with respect to $C \sse V$ if
\begin{enumerate}
    \item $C$ is fully symmetric in $G$,
    \item if $\lambda(e) \neq 0$, then $e$ is completely contained either in $C$ or in $\overline{C}$ (\ie every edge $e$ is fully contained either in $C$ or in $\overline{C}$,
    \item if $\lambda(e) \neq 0$ and $e\sse \ov C$, then $\abs{e} = 1$ (\ie every edge $e$ with $e \cap C = \emptyset$ satisfies $|e| = 1$).
\end{enumerate}
Moreover, if the graph satisfies only conditions 1 and 2 we will call  it \emph{partially symmetry purified}.
We write $\symp{V}{C}$ for the set of all symmetry-purified $d$-hypergraphs with respect to $C \sse V$ and $\psymp{V}{C}$ for the set of partially symmetry-purified $d$-hypergraphs with respect to $C \sse V$ (note that $d$ and $p$ are implicitly defined from the context for $\symp{V}{C}$ and $\psymp{V}{C}$). 
\end{definition}



The next crucial lemma allows us to restrict our attention only to expressions  $\po s(G;u)$ over symmetry-purified graphs, which have a very regular and much easier to analyze structure. This enables a later combinatorial analysis of the periodic behaviour of such $\po  s(G;u)$. The proof of the lemma relies on carefully applying both \cref{ddl} and \cref{sddl} to alter the graphs while preserving the symmetry of the corresponding expression.

\newcommand{\Hset}{\cH}

\begin{lemma}\label{sym-expr}
Let $p\neq q$ be prime numbers, let $u \in \GF{p}$, and let $G =(V,\lambda)$ be a $\GF{p}$-labeled $d$-hypergraph. Moreover, let $C\sse V$ be a maximal fully symmetric subset with $\abs{C} > \abs{V}/2$.
Then there are constants $\beta_{H, r} \in \F_q$ such that 
\[\po s(G;u)\equiv \sum_{H \in \symp{V}{C}} \sum_{r\in \GF{p}} \beta_{H, r}  \,\po s(H, r).\]
\end{lemma}

\paragraph*{Period of symmetry-purified expressions.}

For a fixed input $\ov b \in \{0,1\}^n$ and an $n$-ary symmetric expression $\po f$, we can compute the value $\po f(\ov b)$ only knowing the hamming weight of the input, i.e.\ the number of $1$s in $\ov b$. 
This means that $\po f$ represents not only a function $\{0,1\}^n \rightarrow D$, but we can also view it as a function $\{0, 1, \ldots, n\} \rightarrow D$.
It turns out that a relatively short $\cdhexpr{d}{p}{q}$ circuit can compute only functions with a relatively small period.
Here, by a period of $\po f$ we mean an integer $r \in \N\setminus\{0\}$ which satisfies $\po f(m+r) = \po f(m)$ for all $m$ in range $[0, n-r]$. Note that all the functions have periods $>n$, so we are mainly interested in finding periods in range $[1,n]$.
Note that the $\cAND{n}$ function, which is of our particular interest, does not have any period (less than $n+1$). 
Thus proving an upper bound for a period of a relatively short symmetric function will give us a lower bound for the length of representation of $\cAND{n}$.
This is in line with some of the previous research \cite{BarringtonBR94,Grolmusz01,straubing2006note}. 
As any $\cdhexpr{d}{p}{q}$ can be transformed to a symmetric expression over symmetry-purified graphs, we only need to concentrate on these special graphs. Indeed, any common period among all the elements of the sum, transfers to the sum itself.
For the following theorem, we need a careful analysis how an expression $\po s(G;u)$ for some symmetry purified graph $G$ is computed. We rely on the fact that, for fixed $s \in \N$, the function $m \mapsto \binom{m}{s}\bmod p$ is periodic with period $p^{k}$ for each $k\in \N$ such that $p^k > s$ (see for instance \cite[Proof of Fact 3.4]{IdziakKK22LICS}).

\newcommand{\acpt}{r}

\begin{theorem}\label{period-sym-pure}
Let $p\neq q$ be prime numbers, $\acpt \in  \GF{p}$ and let $G$ be a $\GF{p}$-labeled $d$-hypergraph on $n$-vertices that is symmetry purified with respect to a maximal fully symmetric subset of vertices $C$ of size $|C|> n/2$.
Then $p^{k_p} \cdot q^{k_q}$ is a period of $\po s(G; \acpt)$ where $k_p$ is the smallest integer satisfying $p^{k_p} > d$ and $k_q$ is the smallest integer satisfying $q^{k_q} > n-|C|$. 
\end{theorem}
\newcommand{\grp}{\Gamma}

\paragraph*{Main theorems.}
Now we have all necessary components to prove our main theorems.

\begin{proof}[Proof of \cref{thm:mainPeriod}]
 As discussed in \cref{preparation}, any $\cdhcirc{d}{p}{q}$ circuit has a corresponding symmetric $\cdhexpr{d}{p}{q}$ expression $\po f$ with no periods smaller than the circuit we started with (note that there can happen some blow-up in size, but this does not matter as we argue below).
By \cref{sym-sum}, $\po f$ can be written as a sum of expressions of the form $\po s(G;\acpt)$.
Hence, from now on let us consider one of these expressions $\po s(G;\acpt)$. 

We choose $\epsilon$ such that  $ 2 s  = 2^{\eps \cdot n}$  (meaning that $\eps \cdot n = \log s + 1$ and $\eps < 1/8$).
If $G$ does not contain a fully symmetric set $|C|$ of size at least $n - \floor{\eps n}$, by \cref{automo-lem}, it satisfies $\abs{\sym{[n]}/\aut{G}} \geq 2^{\floor{\eps n}}$.
Thus, writing $\po s(G; r) = \sum_{i=0}^k \po b^{\pi_i}(G, r)$ as in \cref{decomp-sum}, it follows that  $k \geq 2^{\floor{\eps n}}$. As all the different terms in this sum get their inputs from different graphs $\pi_i(G)$, also for each term in the sum there must have been a different gate in the original circuit we started with.
This is a contradiction as $2^{\floor{\eps n}} > s$.

Hence, all the subexpressions $\po s(G;\acpt)$ contain a fully symmetric set $C$ of size at least $n - \floor{\eps n}$. Now, \cref{sym-expr} tells us that we can write $\po f$ as a sum of expressions of the form $\po s(G;\acpt)$ where $G$ is symmetry-purified with respect to $C$.  Then, \cref{period-sym-pure} implies that each such $\po s(G; \acpt)$ has period $p^{k_{p}}\cdot q^{k_q}$, where $k_{p}$ is the smallest integer such that $p^{k_{p}} > d$ and $k_q$ is the smallest integer such that $q^{k_q} > n-|C| = \floor{\eps n} $. 
As $\po f$ is a sum of different $\po s(G; \acpt)$, which all share the period $p^{k_{p}}\cdot q^{k_q}$, it itself has period $p^{k_{p}}\cdot q^{k_q}$.
\end{proof}

Instead of directly proving \cref{thm:mainGeneralIntro}, let us derive the following slightly more explicit and general variant of the theorem:

\begin{theorem}\label{thm:mainGeneralProofs}
   Let $p$ and $q$ be primes and let $n\geq \max\{13, 4p^2q^2\}$ and $d \leq n - \sqrt{n}$. Then every symmetric  $\cdhcirc{d}{p}{q}$ circuit computing the $\cAND{n}$ function has size at least $2^{ \max\{n/(2dpq), \sqrt{n}\}}$.
\end{theorem}

\begin{proof}
Let us write $V = [n] = \{x_1, \dots, x_n\}$. 
  First consider the case that $d \geq \sqrt{n}$ and there is actually an $\cAND{k}$ gate $v$ with $n-\sqrt{n}\geq k\geq\sqrt{n}$ inputs. Since for any $\pi\in \sym{V}$ also $\pi(v)$ must be a gate in the circuit, we obtain different $\cAND{k}$ gates for each $k$-element subset of $V$. As there are $\binom{n}{k} \geq \max\{(n/k)^k , (n/(n-k))^{n-k}\}\geq 2^{\sqrt{n}} \geq 2^{n/d}$ many such subsets, the theorem holds in this case (almost trivially).

Therefore, in the following, we assume $d < \sqrt{n} \leq n/(2pq)$ and consider an arbitrary $n$-input $\cdhcirc{d}{p}{q}$ circuit $\cC$ of size $s \leq  2^{n/(2pqd)}$. 
 By \cref{thm:mainPeriod}, the function computed by $\cC$  has period $p^{k_{p}}q^{k_q}$ where $k_{p}$ is the smallest integer with $p^{k_{p}} > d$ and $k_q$ is the smallest integer with $q^{k_q} > \log s +1 \geq n/(2pqd) + 1$. Notice that $p^{k_{p}} \leq d \cdot p$ and $q^{k_q} \leq (n /(2pqd) + 1) \cdot q$ and, hence, we have \[p^{k_{p}}\cdot q^{k_q}  \leq d \cdot p \cdot (n /(2pqd) + 1) \cdot q = n/2 + dpq < n.\]
 Thus, $\cC$ does not compute the $\cAND{n}$ function as $\cAND{n}$ does not have any non-trivial period.  
\end{proof}

\section{Further Perspectives}\label{further-remarks}
  Arguably the strongest application of the Degree Decreasing Lemma is Theorem 4 in \cite{Grolmusz01}. It implies that, if all the polynomials over $\GF{p}$ that compose the top levels of a $\cdhcirc{d}{p}{q}$-circuit can be written with a sublinear number of (binary) multiplications, then the circuit can be replaced with a $\modmod{p}{q}$ circuit with only a subexponential blow-up in size. We argue that this kind of theorem cannot be applied in the context of our proof.

Note that a large pseudo-clique in the symmetry-purified expressions are (arbitrary) symmetric polynomials. Most  of symmetric polynomials over $\GF{p}$ require at least a linear number of multiplications in any formula (circuit) defining them. To see it, consider the example $p(\ov x) = \sum_{i<j} x_i \cdot x_j$ as a polynomial over $\F_2$. One can easily check that it represents a function with smallest period $4$. But now, if it could be written with a sub-linear number of multiplications, by Theorem 4 in \cite{Grolmusz01}, it could be represented by a sub-exponential size $\modmod{2}{3}$ circuit. However, this contradicts \cite[Theorem 2.4]{GrolmuszT00} as subexponential size $\modmod{2}{3}$ circuits can only represent periodic functions with period of the form $2\cdot 3^{k}$. This shows that that the Degree Decreasing Lemma cannot be used in this context, as it puts the limitations on its own applicability, by providing arithmetic circuit lower bounds. Such lower bounds can be proved for all non-trivial symmetric polynomials over $\F_p$ with $d\geq p$ using a similar period analysis.
Thus, our symmetry purification technique as well as combinatorial analysis contained in the proof of \cref{period-sym-pure} constitute a substantial improvement over the Degree Decreasing Lemma and its accompanying techniques. 

The results of the present paper indicate what type of symmetric functions might be computable by small, but not necessarily symmetric, $\cdhcirc{d}{p}{q}$ circuits. It is natural to believe that the optimal (or nearly optimal) representation of the symmetric function should also be symmetric. Thus we state the following

\begin{conj}
For fixed $d,p,q$, the only symmetric functions that can be represented by $\cdhcirc{d}{p}{q}$ circuits of subexponential size have to be periodic with some period of the form $p^{k_p} \cdot q^{k_q}$, for $k_p$ being the smallest integer with $p^{k_p}\geq d$ and $k_q$ be such that $p^{k_p} \cdot q^{k_q} \leq n$.
\end{conj}

\section{Proofs}\label{proofs-sec}
\subsection*{Symmetries of hypergraphs}

Recall that for an $\GF{p}$-labeled hypergraph $G = (V,\lambda)$ a subset $C \sse V$ is \emph{fully symmetric} if all $e_1, e_2\subseteq V$ with $|e_1| = |e_2|$  and $e_1 \cap \overline{C} = e_2 \cap \overline{C}$ satisfy $\lambda(e_1) = \lambda(e_2)$.

\begin{fact}\label{lem:fullysymauto}
Let $G$ be an $\GF{p}$-labeled hypergraph $G = (V,\lambda)$.
\begin{itemize}
    \item A subset $C \sse V$ is fully symmetric if and only if $\sym{C} \leq \aut{G}$.
    \item If $C, D\sse V$ are fully symmetric sets with $C\cap D\neq \emptyset$, then also $C \cup D$.
    \item If $C\sse V$ is a maximal fully symmetric set with $\abs{C} > \abs{V}/2$, then $\aut{G} = \sym{C} \times \Gamma$ for some $\Gamma \leq \sym{\ov C} $.
\end{itemize}
\end{fact}

\begin{proof}
    Let $C \sse V$ be fully symmetric and let $\pi \in \sym{C}$ be any permutation on $C$. Then, clearly, for any $e \in \cP(V)$, we have $\abs{e} = \abs{\pi(e)}$ and $e \cap \ov C = \pi(e) \cap \ov C$, meaning that, by definition, $\lambda(e) = \lambda(\pi(e))$. Therefore, $\pi \in \aut{G}$ and, hence, $ \sym{C} \leq \aut{G}$.
    On the other hand, if $ \sym{C} \leq \aut{G}$ and $e_1, e_2\subseteq V$ with $|e_1| = |e_2|$  and $e_1 \cap \overline{C} = e_2 \cap \overline{C}$, then there is a permutation $\pi \in \sym{C}$ with $\pi(e_1\cap C) = e_2\cap C$. As $\pi$ does not move any elements outside $C$, we have $\pi(e_1) = e_2$ and, hence, $\lambda(e_1) = \lambda(e_2)$.

 To see the second point, observe that if $C, D\sse V$ are fully symmetric sets, then $ \sym{C} \leq \aut{G}$ and $ \sym{D} \leq \aut{G}$. As $C\cap D\neq \emptyset$, also $\sym{C \cup D} \leq \aut{G}$ (this can be seen \eg using "bubble sort").

For the third point, observe that for every $\pi \in \aut{G}$ we know that $\pi(C)$ is a fully symmetric set with $C \cap \pi(C) \neq \emptyset$ (because of the size bound $\abs{C} > \abs{V}/2$). Thus, $C \cup \pi(C)$ is also fully symmetric. By the maximality of $C$, it follows that $\pi$ must leave $C$ invariant. Thus, $C$ is closed under the action of $\aut{G}$ and, hence, $\aut{G} \leq \sym{C} \times \sym{\ov C}$. Since also $\sym{C} \leq \aut{G}$, the third point follows.
\end{proof}

Now, we are ready to prove the key lemma, \cref{automo-lem}, which allows us to restrict our attention only to very regular hypergraphs.

\begin{lemma}[\cref{automo-lem}, restated]
Let $0<\eps < 1/8$. Every $\GF{p}$-labeled hypergraph $G = (V,\lambda)$ with $n =\abs{V} \geq 13$ has either
\begin{itemize}
    \item a fully symmetric subset on at least $n - \floor{\eps n}$ vertices, or
    \item its automorphism group satisfies $|\sym{V}/\aut{G})| > 2^{\floor{\eps n}}$. 
\end{itemize}
\end{lemma}
\begin{proof}
    Let us write $\Gamma = \aut{G}$. We say that $\Gamma$ is \emph{small} if $\abs{\Gamma} \leq n!/2^{\floor{\eps n}}$. 
    Let us first show that either $\Gamma$ is small or $G$ contains a pseudo-clique on at least $n - \floor{\eps n}$ vertices (in a second step, we show that this pseudo-clique, indeed, is fully symmetric). 
    
	We start by observing that $\Gamma$ is a subgroup of $\sym{k_1} \times\cdots \times \sym{k_m}$, where $k_i$ are the sizes of the orbits of the action on $G$. If $k_i < n - \floor{\eps n}$ for all $i$, then $\abs{\Gamma} < (n - \floor{\eps n})! \cdot \floor{\eps n}!$ and
  \begin{align*}
     \frac{n!}{\abs{\Gamma}} &> \frac{n!}{(n - \floor{\eps n})! \cdot \floor{\eps n}!} = \binom{n}{\floor{\eps n}} 
           \geq \left(\frac{n}{\floor{\eps n}} \right)^{\floor{\eps n}} > 2^{\floor{\eps n}};
 \end{align*} 
 hence, $\Gamma$ is small.
	
	So from now on there is an orbit $C \sse V$ consisting of at least $n - \floor{\eps n}$ vertices. Recall that $\Gamma \leq \sym{C} \times \sym{\ov C}$ and denote by $\phi : \Gamma\to \sym{C}$ the projection to the first coordinate.

 Suppose $\tilde \Gamma = \phi(\Gamma)$ does not act primitively on $C$ meaning that there is an $\tilde \Gamma$-invariant partition of $C$ with $r$ classes each of which consists of $1  < m < \abs{C}$ vertices (as $\tilde \Gamma$ acts transitively on $C$, it must acts transitively on the classes; hence, they all have the same size). Thus, $\tilde \Gamma$ is isomorphic to a subgroup of the wreath product $\sym{m} \wr \sym{r}$ with $rm = \abs{C}$ (see \cite[Theorem 1.8] {cameron99}). So 
   \begin{align*}
     \frac{\abs{C}!}{|\tilde\Gamma| }&\geq \frac{\abs{C}!}{(m!)^r \cdot r!} = \frac{1\cdots m \cdots \qquad \cdots \abs{C}}{(1\cdots m )\cdot 1 \cdot (1 \cdots m )\cdot 2 \cdots ( 1 \cdots m) \cdot r} = \frac{\prod_{i=1}^{r-1}\prod_{j=1}^{m-1} (im+j)}{((m-1)!)^{r-1}}\\[2mm]
            &\geq 2^{(m-1)(r-1)}\geq 2^{\abs{C}/4 }\geq 2^{\floor{\eps n}}.
            \end{align*}
            Here the last inequality is because $\eps < 1/8$ (in particular $1-\eps \geq 1/2)$, the second last inequality is due to the assumption $\eps \leq 1/4$ and the fact that $m-1\geq m/2$ and $r-1\geq r/2$. The third last inequality is because $\bigl(\prod_{j=1}^{m-1}(im+j)\bigr) /(m-1)! =  \prod_{j=1}^{m-1}(im+j)/j \geq 2^{m-1}$ as $i\geq 1$.
Since the index of $\Gamma$ in $\sym{V}$ is at least the index of $\tilde \Gamma = \phi(\Gamma)$ in $\sym{C}$, again $\Gamma$ is small.
	
	Hence, it remains to consider the case that $\phi(\Gamma) \leq \sym{C}$ acts primitively on $C$. 
 Thus, writing $k = \abs{C},$ according to \cite{Bochert89,PraegerS80} (see also \cite{Babai18}), there are three possibilities: $\phi(\Gamma)$ is either $A_k$ (the alternating group on $k$ elements) or $S_k \cong\sym{C}$ or $\abs{\phi(\Gamma)} \leq 4^k$. First, let us consider the last case. As $n \geq 13$, we have $k\geq n - \floor{\eps n} \geq 11$. Therefore, we conclude that we have $\abs{\phi(\Gamma)} \leq 4^k \leq k! / 2^{k/4}$  (which holds for all $k \geq 11$) and, as above, the index of $\Gamma$ in $\sym{V}$ is at least $2^{k/4} \geq 2^{(n - \floor{\eps n})/4} \geq 2^{\floor{\eps n}}$,  meaning that $\Gamma$ is small. 
 
 In the former two cases (\ie that $\phi(\Gamma)$ is $A_k$ or $S_k$), $\phi(\Gamma)$ acts set-transitively on $C$ (meaning that for each pair of subsets $A,B \sse C$ with $\abs{A} = \abs{B}$ there is a permutation $\pi$ mapping $A$ to $B$); hence, $C$ is a pseudo-clique. 

To see that $C$ is, indeed, fully symmetric if $\Gamma$ is not small, we proceed as follows: 
Now, let $N\leq \Gamma$ denote the kernel of the projection $\Gamma\to  \sym{\ov C}$  to the second component (i.e.\ the pointwise stabilizer of $\overline C$). Then we have $\phi(N) = N$ (when identifying $\sym{ C}$ with the corresponding subgroup of $ \sym{C} \times \sym{\ov C} $). As $A_k$ is simple and the index of $N$ is at most $(n-k)!$ in $\phi(\Gamma)$ (which is either $A_k$ or $S_k$), we have $N = A_k$ or $N = S_k$ (as $N$ is also normal in $\phi(\Gamma)$). In both cases $N$ acts set-transitively on $C$ and; hence, $C$ is fully symmetric (which, by \cref{lem:fullysymauto}, also excludes the case $N = A_k$).
\end{proof}

\subsection*{Reduction based on the Degree Decreasing Lemma}


\begin{lemma}[Symmetric Degree Decreasing Lemma~-- \cref{sddl} restated]
Let $p\neq q$ be prime. Let $\gamma \in \GF{p}\setminus \{0\}$. Then every function $f:\{0,1\}^d \times \GF{p} \mapsto \GF{q}$ represented by a $d+1$-ary $\cdhexpr{d}{p}{q}$ expression
\[\po b(\gamma \cdot x_1 \cdot \ldots \cdot x_d + y; t)\]
can be also represented by expression
 \[\po h(\ov x, y;t) = \po b(y; t) + \sum_{r \in \GF{p}} \beta_{t,r} \sum_{S \subseteq [d]} \alpha_{|S|} \cdot \po  b(\gamma\cdot \Sigma(S)  + y; r)\]
for $\alpha_{|S|} = (-1)^{|S|}$ and some coefficients $\beta_{t,r} \in \F_q$.
\end{lemma}
\begin{proof}
    First, let us prove the lemma for the case when $d=q^k$ for some $k$, with additional assumption that $d \equiv 1 \mod p$ (we can always make $k$ larger to satisfy also the second condition).
     We claim that, in this case, the following coefficients do the job: put $\beta_{t,r} := 0$ for $r = t$ and $\beta_{t,r} := 1$ for $r \neq t$. Then, the formula for $\po h$ takes the form
$$\po b(y; t) + \sum_{r \in \GF{p} \setminus \{t\}}\sum_{S \subseteq [d]} (-1)^{|S|} \cdot \po  b(\gamma\cdot \Sigma(S)  + y; r)$$
Now, if any variable $x_i$ is set to $0$ inside $\po h$, then the expressions   $\gamma\cdot \setsum{S\cup\{i\}} + y$ and $\gamma\cdot\setsum{S} + y$ always evaluate to the same value (no matter what the other variables are); hence, each $\po b(\gamma\cdot \setsum{S\cup\{i\}} + y;r)$ cancels out with $\po b(\gamma\cdot\setsum{S} + y;r))$ as $\alpha_{|S|} = - \alpha_{|S|+1}$. 
 After all the cancellations, we are left with $\po b(y; t)$. On the other hand,  $\po b(\gamma \cdot x_1 \cdot \ldots \cdot x_d + y; t)$ also turn into $\po b(y; r)$ after substituting zero for any $x_i$; hence the corresponding functions return the same value in this case. 
 
 Now assume that $x_1 = x_2 = \ldots = x_d = 1$. Then the expression $\po h$ can be rewritten as
\[\po b(y; t) + \sum_{r \in \GF{p}\setminus \{t\}} \sum_{i=0}^d  {d \choose i} \cdot (-1)^{i} \cdot \po b(\gamma \cdot i + y;r).\]
Recall that we assumed $d= q^{k}$; so for $0<i<d$, we have ${d \choose i} = 0 \parmod{q}$. Moreover, if $q$ is odd, then $(-1)^d = -1$ and, if $q = 2$, then $1 = -1 \parmod{q}$. Hence, in any case we have $(-1)^d = 1 \parmod{q}$.
Furthermore, as $d = 1 \parmod{p}$, we can simplify the expression to get
\begin{align*}
   \po b(y; t) + \sum_{r \in \GF{p} \setminus \{t\}}\po b(y; r) - \po b(\gamma+y;r). 
\end{align*} 

Next, observe that  $\sum_{r \in \GF{p}} \po b(y;r)$ always evaluates to $1$ and also $\sum_{r \in \GF{p}} \po b(\gamma +y;r)$ always evaluates to $1$; hence the formula simplifies to 
$$\po b(\gamma +y;r).$$
This is exactly what we wanted, as $\po b(\gamma \cdot x_1 \cdot \ldots \cdot x_d + y; t)$ evaluates to the same value after substituting 1 for $x_1, \dots, x_d$. This concludes the proof of the lemma for our special choice for $d$.

 To finish the proof for arbitrary $d$, we argue that, if the statement of the lemma holds for some $d>1$, it must also hold for $d-1$. To avoid confusion, let $\beta_{t,r}^{(d)}$ correspond to the constants that we already computed, which build the expression $\po h^{(d)}$. We want to find $\beta_{t,r}^{(d-1)}$ which correspond to $\po h^{(d-1)}$. We claim that following recursive formula works:
 \begin{equation*}
    \beta_{t,r}^{(d-1)}:= \beta_{t,r}^{(d)} - \beta_{t,r+1}^{(d)}\\ 
 \end{equation*}
(recall that $r \in \F_p$, so $p$ is interpreted as $0$). Clearly, if we make a substitution $x_d := 1$ in $\po h^{(d)}$ the newly created formula computes a function of the form $\po b(\gamma \cdot x_1 \cdot \ldots \cdot x_{d-1} + y; t)$.
In the expression $\po h^{(d)}[x_d := 1]$ every summand $\po b(\setsum{S}+y;r)$ of $\po h^{(d)}$ with $d \in S$ turns into $\po b(\setsum{S \setminus \{d\}} + y +1 ;r)$, which is equivalent to $\po b(\setsum{S\setminus \{d\}} + y  ;r-1)$. Hence, $\po h^{(d)}[x_d := 1]$ computes the same function as 
 \[\po b(y; t) + \sum_{r \in \GF{p}} \sum_{S \subseteq [d-1]} (\beta_{t,r}^{(d)} \alpha_{|S|} + \beta_{t,r+1}^{(d)} \alpha_{|S|+1})\cdot \po  b(\gamma\cdot \Sigma(S)  + y; r).\]
Note that, as $\alpha_{|S|} = -\alpha_{|S|+1}$, we have $\beta_{t,r}^{(d)} \alpha_{|S|} + \beta_{t,r+1}^{(d)} \alpha_{|S|+1} = \beta_{t,r}^{(d)} \alpha_{|S|} - \beta_{t,r+1}^{(d)} \alpha_{S} = (\beta_{t,r}^{(d)} - \beta_{t,r+1}^{(d)}) \cdot \alpha_{S}$. Hence, we end up with the formula 
 \[\po b(y; t) + \sum_{r \in \GF{p}} (\beta_{t,r} - \beta_{t,r+1}) \sum_{S \subseteq [d-1]}  \alpha_{|S|}\cdot \po  b(\gamma\cdot \Sigma(S)  + y; r)\]
 which shows that our recursive formula for $\beta_{t,r}^{(d-1)}$ is correct.
\end{proof}

The key property of the formula $\po h$ is that it is invariant to the permutations of variables $x_1, \ldots, x_d$, which was not the case for original Degree Decreasing Lemma of \cite{Grolmusz01}. In this paper we will only need the following application of \cref{sddl}.
\begin{corollary}\label{sddliteration}
Let $p \neq q$ be prime numbers. Let $H$ be $\GF{p}$-labeled $d$-hypergraph, let $e_1 \ldots, e_l\subseteq V$ be some subsets of $V$ of the same size $m$, let $\gamma , t \in \GF{p}$. Then function computed by the expression
$$\po b(\gamma(e_1 + \ldots + e_l) + H; t)$$
can be also computed by
\[
\po s(H;t) \equiv  
 \sum_{\cS : S_i \subseteq e_i} \sum_{r \in \GF{p}} \beta_{\cS}^{(t,r)} \cdot \po  b(\gamma \cdot (\setsum{S_1} + \setsum{S_2} + \ldots + \setsum{S_l}) + H'; r) 
\]
for suitable constants $\beta_{\cS}^{(t,r)} \in \F_q$ where $\cS = (S_1, \ldots, S_l)$ and $\beta_{\cS}^{(t,r)} = \beta_{\sigma(\cS)}^{(t,r)}$ for any permutation $\sigma \in \sym{l}$.
\end{corollary}

\begin{proof}
It is clear that by a repeated application of \cref{sddl}, we will end up with an expression as claimed in the corollary~-- without knowing the symmetry condition. In order to prove that $\beta_{\cS}^{(t,r)} = \beta_{\sigma(\cS)}^{(t,r)}$ for any permutation $\sigma \in \sym{l}$, we proceed as follows:
Let $e_1 = x_1\cdots x_m$ and $e_2 = y_1\cdots y_m$ and consider the expression $\po b(\gamma(e_1 + e_2) + H; t) $. We are going to show that, if we are applying \cref{sddl} first to $e_1$ and then to $e_2$, we get the same result as when we apply \cref{sddl} first to $e_2$ and then to $e_1$. By induction then we can apply any permutation to $\cS$.
Thus, let us apply \cref{sddl} first to $e_1$ and then to all appearances of $e_2$:
\begin{align*}
\po b(\gamma(e_1 + e_2) + H; t)
&= \po b(\gamma(e_2) + H; t) + \sum_{r\in \F_p} \beta_{t,r} \sum_{S \subseteq e_1} \alpha_{|S|} \cdot \po  b(\gamma\cdot (\setsum{S} + e_2) + H; r)\\
& = \po b(H; t)\\
&\qquad + \sum_{r\in \F_p} \beta_{t,r} \sum_{S \subseteq e_2} \alpha_{|S|} \cdot \po  b(\gamma\cdot \setsum{S} + H; r)\\
&\qquad + \sum_{r\in \F_p} \beta_{t,r} \sum_{S \subseteq e_1} \alpha_{|S|} \cdot \po  b(\gamma\cdot \setsum{S} + H; r)\\
&\qquad + \sum_{r,r'\in \F_p} \beta_{t,r}\beta_{r,r'} \sum_{S \subseteq e_1}\sum_{S' \subseteq e_2} \alpha_{|S|}\alpha_{|S'|} \cdot \po  b(\gamma\cdot \setsum{S} + \gamma\cdot \setsum{S'} +  H; r)
\end{align*}
We see that this is exactly the same as if we would \cref{sddl} first to $e_2$ and then to all appearances of $e_1$ (just the second and third line would be exchanged and the sum $\sum_{S \subseteq e_1}\sum_{S' \subseteq e_2}$ replaced by $\sum_{S \subseteq e_2}\sum_{S' \subseteq e_1}$, which, of course, is the same).    
\end{proof}

The next lemma is an application of Symmetric Degree Decreasing Lemma to expressions generated by hypergraphs.
Let us recall that $\symp{V}{C}$ denotes the set of all symmetry-purified $d$-hypergraphs and $\psymp{V}{C}$ the set of partially symmetry-purified $d$-hypergraphs with respect to $C \sse V$.
A $\GF{p}$-labeled hypergraph $G = (V,\lambda)$ is symmetry-purified with respect to $C \sse V$ if $C$ is fully symmetric, every $e\sse V $ with $\lambda(e) \neq 0$ is fully contained either in $C$ or in $\overline{C}$, and every $e \sse \ov C$ with $\lambda(e) \neq 0$  satisfies $|e| = 1$.

\begin{lemma}\label{part-sym}
Let $p\neq q$ be prime numbers, let $u \in \GF{p}$, and let $G = (V, \lambda)$ be a $\GF{p}$-labeled $d$-hypergraph that contains a fully symmetric subset of vertices $C$ with $|C| > \frac{\abs{V}}{2}$. Moreover, let $C$ be a maximal such set. Then the function represented by $\po s(G;u)$ can be also represented by a sum of the form
$$\sum_{H \in \psymp{V}{C}} \sum_{r\in \GF{p}} \beta_{H, r} \, \po s(H; r)$$
\end{lemma}

\begin{proof}

We will use a recursive reasoning, in which we will modify a sum of the form 
$$\sum_{H \in \Hset} \sum_{r\in \GF{p}} \beta_{H, t} \, \po{s}(H; t)$$
 so that we gradually remove edges with non-empty intersection with both $C$ and $\ov C$, while preserving the computed function. 
We start with a sum of length one consisting of the only element $1\cdot s(G; u)$~-- so at the beginning the set $\Hset$ has size one.
In each step we pick a graph $H\in \Hset$ with a maximum number of edges that intersect both $C$ and $\ov C$. Then we will pick an edge together with its orbit under $\aut{H}$ and create from $G$ some number of graphs that we will add to the new $\Hset$ in place of $G$. The added graphs will not contain the picked edge, nor its orbit, or any new edge that intersects with both $C$ and $\overline{C}$ non-trivially. If we successfully do so, we may increase the number of graphs in $\Hset$, but we decrease the number of graphs with the maximum number of edges intersecting both $C$ and $\overline{C}$. So, if we keep repeating this process, eventually there will be no bad graphs in $\Hset$.

So pick a graph $H \in \Hset$ which has some edge $e\in E_H$ with non-empty intersection to both $C$ and $\overline{C}$.
We will denote the corresponding intersections by $e^C = e \cap C$ and $e^{\overline{C}}= e \cap \overline{C}$  respectively.
Let  $ O = \aut{H} \cdot e$ denote the orbit of $e$ and enumerate its edges as $ O = \{e_1, \ldots, e_\ell\}$.
In particular, all the  $e_1, \ldots, e_\ell$ share the same label $\gamma= \lambda(e_i)$.
As $C$ is a maximal fully symmetric set with $|C|>\frac{n}{2}$, by \cref{lem:fullysymauto} we have $\aut{G} = \sym{C} \times \Gamma$ for some $\Gamma \leq \sym{\ov C} $. Setting $P= \aut{H} \cdot e^{C}$ and $Q = \aut{H} \cdot e^{\ov C}$, we can identify $O$ with $P\times Q$. Hence, we obtain
$$\gamma (e_1 + \ldots +e_\ell) = \gamma \cdot (\sum_{w \in P} w)(\sum_{v \in Q} v).$$
Hence, we have 
$$
\po s(H;t) = \po s(\gamma (e_1 + \ldots +e_\ell) + H'; t) = \po s(\gamma  \cdot z_1 \cdot z_2 + H'; t) 
$$
where $z_1 = \sum_{w \in P} w$, $z_2 = \sum_{v \in Q} v$ and $H' = H - \gamma\cdot (e_1 + \ldots + e_l)$.
So, from the definition of $\po s$, we have
$$
\po s(H; t) = \sum_{i=1}^k \po b^{\pi_i}(\gamma \cdot z_1 \cdot z_2 + H' ;t).
$$
 (recall that $\pi_1, \ldots, \pi_k$ is a transversal of $\sym{V}/\aut{G}$). So now by applying \cref{ddl} to each term of the sum individually we get 
$$
\po s(H; t) \equiv \sum_{(j_1, j_2, j_3) \in \GF{p}^3} \sum_{r \in \GF{p}} \sum_{i=1}^k \beta_{j_1, j_2, j_3}^{(r)}  \po b^{\pi_i}(j_1 z_1 + j_2 z_2 + j_3 H'; r)
$$
Note that in the resulting expression we replaced edges from $O$ that have nontrivial intersection with both $C$ and $\ov C$ by edges that are fully contained in either $C$ or $\ov C$. Moreover, the resulting expression is symmetric. Indeed, $z_1$ and $z_2$ by definition are expressions closed under $\aut{H}$. Moreover, $H'$ is closed under automorphisms of $H$, as it was created by removing the entire orbit from $H$. Hence, also each of the expressions of the form $j_1 z_1 + j_2 z_2 + j_3 H'$ is closed under $\aut{H}$ and, thus, for fixed $j_1, j_2, j_3$ and $r$ the sum  
$$\sum_{i=1}^k \po b^{\pi_i}(j_1 z_1 + j_2 z_2 + j_3 H'; r)$$
is symmetric. Indeed, $\aut{H} \leq \aut{H'}$ and the set $\{\pi_1, \dots, \pi_k\}$ contains the same number of representatives for each $\sym{V}/\aut{H'}$ class.

Hence, the new expression replacing $\po s(H; t)$ is also symmetric. This finishes the proof, as now we just regroup the resulting expression.
\end{proof} 

\begin{lemma}[\cref{sym-expr} restated]
Let $p\neq q$ be prime numbers, let $u \in \GF{p}$, and let $G =(V,\lambda)$ be a $\GF{p}$-labeled $d$-hypergraph. Moreover, let $C\sse V$ be a maximal fully symmetric subset with $\abs{C} > \abs{V}/2$.
Then the function represented by $\po s(G;u)$ can be also represented by the sum of the form
\[\sum_{H \in \symp{V}{C}} \sum_{r\in \GF{p}} \beta_{H, r}  \,\po s(H, r)\]
for suitable constants $\beta_{H, r} \in \F_q$.
\end{lemma}

\begin{proof}
By \cref{part-sym} we can assume $G$ is already partially symmetry purified with respect to $C$. We repeat the strategy from the proof of the \cref{part-sym} and slowly modify sum of the form 

\[\sum_{H \in \Hset} \sum_{t \in \GF{p}} \beta_{H, t} \,\po s(H; t)\]
so that we gradually remove edges of degree greater than $1$ to make $\Hset$ contain only symmetry-purified graphs. 
Obviously, we start with a sum of length one with $\Hset= \{1\cdot s_{G, u}\}$. 
All the graphs that we will add to $\Hset$ in the future will be partially symmetry purified. Now, our strategy here is similar to the one used in \cref{part-sym}. First let $m = \max\{|e|: (V,\lambda) \in \Hset,  e\subseteq \ov C, \lambda(e) \neq 0\}$.
Now pick the graph $H$ that has the biggest number of edges realizing this maximum $m$. We will replace $H$ with some number of graphs that have less edges of size $m$. So eventually we will end-up with no edges of size $>1$ and the resulting $\Hset$ will contain only symmetry-purified graphs.

Now, pick a graph $H = (V, \lambda)\in \Hset$ and some $e\subseteq \ov C$ of size $m$ with $\lambda(e) \neq 0$.
We are going to apply \cref{sddliteration} to the edges in the orbit $\aut {H}\cdot e$ \footnote{If we applied \cref{sddl} only to the edge $e$ we could destroy symmetry.} (where the elements of the orbit by definition must have the same label, as $\aut {H}$ is the group of automorphisms of labeled graph).
Let $\aut {H}\cdot e = \{e_1, e_2, \ldots, e_l\}$  and write $\gamma = \lambda(e_i)$. Hence,

$$
\po s(H,t) = \sum_{i=1}^k \po b^{\pi_i}(\gamma \cdot e_1 + \ldots + \gamma \cdot e_l + H'; t)
$$
where $H'$ is created from $H$ by removing the edges $(e_j)_{1..l}$.

Now, write $\cQ = \cP(e_1) \times \dots \times \cP(e_l)$ and apply \cref{sddliteration}, to get 
\[
\po s(H;t) \equiv  
 \sum_{\cS \in \cQ} \sum_{r \in \GF{p}} \sum_{i=1}^k \beta_{\cS}^{t,r} \cdot \po  b^{\pi_i}(\gamma \cdot (\setsum{S_1} + \setsum{S_2} + \ldots + \setsum{S_l}) + H'; r) = \po s'(H; t) 
\]
 for proper $\beta_{\cS}^{t,r} \in \F_q$ such that $\beta_{\cS}^{(t,r)}$ depends only on $(t,r)$ and the multi-set of values $\{\!\!\{|S_1|, \ldots, |S_l|\}\!\!\}$  where $\cS = (S_1, \ldots, S_l)$ (in other words, as stated in \cref{sddliteration}, permuting the elements of $\cS$ does not change  $\beta_{\cS}^{(t,r)}$). 

\newcommand{\class}{Q}
Let $\sim$ denote the equivalence relation on elements of $\cQ$ that induce the same multiset $\{\!\!\{|S_1|, \ldots, |S_l|\}\!\!\}$ (\ie $\cS \sim \cS'$ if and only if there is some $\sigma \in \sym{l}$ with $\sigma(\cS) = \cS'$).
We re\-group the above sum into subsums according to $r$ and according to the $\sim$-class $\class$ in order to obtain
$$\po s'(H,t) = \sum_{\class \in \cQ/\!\sim} \; \sum_{r \in \GF{p}}\beta_{\cS}^{t,r} \po s_{\class}(H'; r)$$
where 
$$
\po s_{\class}(H'; r) = \sum_{(S_1, \ldots, S_l) \in \class} \sum_{i=1}^k \po b^{\pi_i} (\gamma \cdot (\setsum{S_1} + \setsum{S_2} + \ldots + \setsum{S_l}) +H'; r).
$$

We will prove that such $\po s_{\class}(H'; V)$ is a symmetric expression. This implies that $\po s'(H; t)$ is also symmetric (as a linear combination of symmetric expressions) and, hence, concludes the proof.

So let $\sigma \in \sym{V}$ denote any permutation of the vertex set $V$. We want to prove that $\sigma\bigl(\po s_{\class}(H'; r)\bigr) = \po s_{\class}(H'; r)$, where
$$
\sigma\bigl(\po s_{\class}(H'; r)\bigr) = \sum_{(S_1, \ldots, S_l) \in \class} \sum_{i=1}^k \po b^{\sigma \circ \pi_i} (\gamma \cdot (\setsum{S_1} + \setsum{S_2} + \ldots + \setsum{S_l}) +H'; r)
$$

Recall that $(\pi_i)_{1..k}$ is a transversal of $\sym{V}/\aut{H}$. This means that the map $[\pi_i] \mapsto \sigma \circ [\pi_i]$ is a permutation of $\aut{H}$ classes. So, for a fixed $\sigma$, there is a unique permutation $\phi \in \sym{k}$ such that $\sigma \circ [\pi_i] = [\pi_{\phi(i)}]$. Hence, we have $\sigma \circ \pi_i =  \pi_{\phi(i)} \circ \sigma'$ for some $\sigma' \in \aut{H}$. As $H'$ is the subgraph of $H$ that was created by removing an entire orbit of edges from $H$, we have $\aut{H} \leq \aut{H'}$. Therefore, we have $(\sigma \circ \pi_i)(H') = \pi_{\phi(i)}(H')$ and we can see that 
\begin{align*}
\sigma\bigl(\po s_{\class}(&H'; r)\bigr)\\&=\sum_{(S_1, \ldots, S_l) \in \class} \sum_{i=1}^k \po b^{\sigma \circ \pi_i} (\gamma \cdot (\setsum{S_1} + \setsum{S_2} + \ldots + \setsum{S_l}) +H'; \, r) \\
 &= \sum_{(S_1, \ldots, S_l) \in \class} \sum_{i=1}^k \po b \Bigl( \gamma \cdot \bigl( \setsum{(\sigma \circ \pi_i)(S_1)} + \ldots + \setsum{(\sigma \circ \pi_i)(S_l)}\bigr) +(\sigma \circ \pi_i)(H'); \, r \Bigr) \\
 &= \sum_{(S_1, \ldots, S_l) \in \class} \sum_{i=1}^k \po b \Bigl(\gamma \cdot \bigl(\setsum{(\pi_{\phi(i)} \circ \sigma')(S_1)} + \ldots + \setsum{(\pi_{\phi(i)} \circ \sigma')(S_l)}\bigr) + \pi_{\phi(i)}(H'); \, r \Bigr) \\
 &= \sum_{(S_1, \ldots, S_l) \in \class} \sum_{i=1}^k \po b^{\pi_{\phi(i)}} \Bigl(\gamma \cdot \bigl(\setsum{\sigma'(S_1)} + \ldots + \setsum{\sigma'(S_l)}\bigr) + H'; \, r\Bigr).
\end{align*}

Now, since $\sigma'$ is an automorphism of $H$, it permutes the edges $e_1, \ldots, e_l$ (as they form an orbit). 
Hence, we conclude that there is a permutation $\psi \in \sym{l}$ such that $(S_1, \ldots, S_l) \in \class$ if and only if $(\sigma'(S_{\psi(1)}), \ldots, \sigma'(S_{\psi(l)})) \in \class$ (as the multisets $\{\!\!\{|S_1|, \ldots, |S_l|\}\!\!\}$ and $\{\!\!\{|\sigma'(S_1)|, \ldots, |\sigma'(S_l)|\}\!\!\}$ coincide). Hence, applying $\sigma'$ to the $S_i$ only changes the order of the outer sum and we obtain that
\begin{align*}
\sigma\bigl(\po s_{\class}(H'; r)\bigr)=\sum_{(S_1, \ldots, S_l) \in \class} \sum_{i=1}^k \po b^{\pi_{\phi(i)}} (\setsum{S_1} + \ldots + \setsum{S_l} + H'; r)
\end{align*}
Finally, we observe that this sum is identical to $\po s_{\class}(H'; r)$ except we changed the order of summation in the second sum. Hence, 
$\sigma(\po s_{\class}(H'; r)) = \po s_{\class}(H'; r)$,
which shows that $\po s_{\class}(H'; r)$ is symmetric. 
\end{proof}

\subsection*{Ordered partitions and multinomial coefficients}

In order to show \cref{period-sym-pure}, we need some knowledge about the periodicity of binomial coefficients.
For a fixed $s$, we can consider binomial coefficient $m \choose s$ as a function of integer $m\geq 0$. It is a known fact that such a function, when considered modulo a prime $p$, is periodic  with period $p^{k}$ for each $k\in \N$ such that $p^k > s$ (see for instance \cite[Proof of Fact 3.4]{IdziakKK22LICS}). It is interesting to see, whether this periodicity generalizes also to multinomial coefficients. Recall that a multinomial coefficient ${n \choose s_1, \ldots, s_l}$ counts the number of ordered partition of $n$ elements set into sequences of disjoint subsets of sizes $s_1, \ldots, s_l$. More formally, for $s_1 + \dots + s_l = n$ we have
\begin{align}
 {n \choose s_1, \ldots, s_l} = {s_1 \choose s_1} \cdot {s_1 + s_2 \choose s_2} \cdot \ldots \cdot {s_1 + \ldots + s_l \choose s_l}. \label{multinomial}
\end{align}

\begin{lemma}\label{period-newton}
    Let $p$ be a prime. Let $s_1, \ldots s_{l-1}$ be a sequence of integers. Let $k$ be such that $p^k > s_1 +\ldots + s_{l-1}$. Then the function
    $$f(n) = {n \choose s_1, \ldots, s_l(n)} \mod p$$
    is periodic with period $p^k$ where $s_l(n) = n-(s_1 + s_2 + \ldots s_{l-1})$.
\end{lemma}
\begin{proof}
We use the following formula \eqref{multinomial} for computing multinomial coefficient.
Note that the first $l-1$ factors of the above product are constants, while the last one satisfies  
$${s_1 + \ldots + s_l \choose s_l(n)} = {n \choose 
s_{l}} = {n \choose n-s_l} = {n \choose s_1 +\dots +s_{l-1}}.$$
\noindent So, periodicity of multinomial coefficient comes from periodicity of the binomial coefficients.
\end{proof}

\subsection*{Period of symmetry-purified expressions}

\begin{theorem}[\cref{period-sym-pure} restated]
Let $p\neq q$ be prime numbers, $\acpt \in  \GF{p}$ and let $G$ be a $\GF{p}$-labeled $d$-hypergraph on $n$-vertices that is symmetry purified with respect to a maximal fully symmetric subset of vertices $C$ of size $|C|> n/2$. 
Then, $p^{k_p} \cdot q^{k_q}$ is a period of $\po s(G; \acpt)$ where $k_p$ is the smallest integer satisfying $p^{k_p} > d$ and $k_q$ is the smallest integer satisfying $q^{k_q} > n-|C|$. 
\end{theorem}
Note that, if  $p^{k_p} \cdot q^{k_q}>n$, it means that we have no non-trivial periods.
\begin{proof}
Note that an induced graph on the set $C$ is a pseudo-clique. Moreover, as $|C| > \frac{|V|}{2}$, we can reconstruct the graph $G$ up to isomorphism having the following information
\begin{enumerate}
    \item the size $l_{C}$ of the largest pseudo-clique $C$ in $G$,
    \item the type $\ov t \in \F_p^d$ of the pseudo-clique, where $t_i$ is the label of every $i$-ary edge in the pseudo-clique
    \item sizes $l_0, l_1, \ldots, l_{p-1}$, where $l_i$ informs how many vertices $j$, which are not in the pseudo-clique  $C$, have a unary edge with label $i=\lambda({\{j\}})$. Vertices corresponding to label $i$ in $G$ will be denoted $L_i$.
\end{enumerate}

Using this characterization of a symmetry purified $G=(V,\lambda)$, we immediately get that
\begin{fact}
\hfil$
\aut{G} = \sym{C} \times \sym{L_0} \times\cdots\times \sym{L_{p-1}}
$\hfill
\end{fact}
In particular, we have at most $p+1$ orbits under the automorphism groups (note that we have less than $p+1$ orbits if some of the $l_i$ are $0$)

Now we evaluate $\po s(G;\acpt)$ on some integer $m$ (and denote this by $\po s(G;\acpt)(m)$).  In order to do it, pick $\ov b$ which has ones on the first $m$ coordinates, i.e $\ov b = 1^{m} \cdot 0 ^{n-m}$. Hence,

$$\po s(G;\acpt)(m) = \sum_{i=1}^k \po b (\pi_i(G);\acpt)(m)$$
where, as before $\pi_1, \ldots, \pi_k$ is a transversal of $\sym{V}/\aut{G}$.
Each summand $\po b (\pi_i(G);\acpt)(m)$ evaluates to 1 or 0, depending on the mapping $\pi_i$. Being more precise, if $\pi_i$ maps $s_0$ elements of $L_0$ to $[1..m]$, $\ldots$, $s_{p-1}$ elements of $L_{p-1}$ to $[1..m]$, then $\pi_i(G) (\ov b)$ evaluates to 
\begin{equation}\label{chisum}
 \sum_{j=1}^d t_j \cdot {s_C(m) \choose j} + \sum_{j=1}^{q-1} j \cdot s_j\ (\mbox{mod}\ p).
\end{equation}
where $s_C(m)$ denotes the number of elements of $C$ mapped to $[1..m]$, which is computed according to formula $s_C(m) = m - (s_0 + \ldots + s_{p-1})$.
Recall that $\po b (\pi_i(G);\acpt)(m) =  1$ iff  $\pi_i(G) (\ov b) =  \acpt$. 

Let $\chi[G; \acpt](m)$ denote the set of $\ov s = (s_0, \ldots, s_{p-1})$ that make the sum \eqref{chisum} evaluate to $\acpt$.
Observe that we have natural inequalities $0 \leq s_i \leq l_i$ for all $i \in \{0, \ldots, p-1\}$ and $0 \leq s_C(m) \leq l_C$. Hence, the feasible $\ov s \in \N^{p}$ can be described by
\begin{equation}\label{chicondition}
\ov s \in \chi[G; \acpt](m)  \iff
\begin{cases}
\sum_{j=1}^d t_j \cdot {s_C(m) \choose j} + \sum_{j=1}^{p-1} j \cdot s_j\ (\mbox{mod}\ p) = \acpt\\
0 \leq s_i \leq l_i \mbox{ for } i \in \{ 0, \ldots, p-1\} \\
0 \leq s_C(m) \leq l_C.
\end{cases}
\end{equation}

Moreover, let $\#[\ov s](m)$ denote the number of permutations in $\{\pi_1, \ldots, \pi_k\}$ that map $s_C(m)$ elements of $C$ to $[m]$ and $s_i$ elements of $L_i$ to $[m]$ (for $i=0,..,p-1$). Hence, we have 
\[\po s[G;\acpt](m) = \left(\sum_{(\ov s) \in \chi[G; \acpt](m)} \#[\ov s](m)\right) (\mbox{mod } q).\]

Next, let us determine $\#[\ov s](m)$. Note that each permutation $\pi_i$ in $\{\pi_1, \ldots, \pi_k\}$ maps each of the sets $L_0, \ldots, L_{p-1}, C$ to some subsets $L_0', \ldots, L_{p-1}', C' \subseteq [n]$ with $|L_0'| = |L_0|, \ldots, |L_{p-1}'| = |L_{p-1}|, |C'| = |C|$.
Now, if two mappings $\pi_i, \pi_j$ output the same image $\pi_i L_0 = \pi_j L_0, \ldots, \pi_{i} L_{p-1} = \pi_j L_{p-1}, \pi_i C = \pi_j C$, then $\pi_i G = \pi_j G$ and hence $\pi_i$ and $\pi_j$ must belong to the same coset of $\aut{G}$ in $\sym{n}n$. Since $\{\pi_1, \ldots, \pi_k\}$ were chosen to be a transversal of $\sym{n}/\aut{G}$, this is only possible when $\pi_i = \pi_j$. So, the mapping $\pi_i \mapsto (L_0', \ldots, L_{p-1}', C')$ is injective. 

In fact, the mapping is also surjective as for any particular $(L_0', \ldots, L_{p-1}', C')$ with $|L_0'| = |L_0|, \ldots, |L_{p-1}'| = |L_{p-1}|, |C'| = |C|$ we can find some $\pi \in \{\pi_1, \ldots, \pi_k\}$ which maps $\pi(L_{i}) = L_i'$ and $\pi(C) = C'$. Indeed, just pick any $\sigma \in \sym{n}$ satisfying $\sigma(L_{i}) = L_i'$ for all $i$ and $\sigma(C) = C'$ and take its representative in the equivalence class modulo $\aut{G}$ as $\pi_i$.
So we have a bijective mapping between permutations $(\pi_i)_{i=1..p-1}$ and ordered partitions $(L_0', \ldots, L_{p-1}', C')$ of $[n]$ satisfying $|L_{0}'| = |L_0| \ldots, |L_{p-1}'| = |L_{p-1}|, |C'| = |C|$. 

 Thus, if we want to count $\#(\ov s)$,  we need to count the number of proper partitions satisfying $|L_1' \cap [1..m]| = s_1, \ldots, |L_{p-1}' \cap [1..m]| = s_{p-1}$, and $ |C' \cap [1..m]| = s_C(m)$.
In other words, we partition an $m$-element set into disjoint subsets of sizes $s_0, s_1, \ldots, s_{p-1}, s_C(m)$ and an $n-m$-element set into disjoint subsets of sizes $l_0-s_0, \ldots, l_{p-1} - s_{p-1}, l_C - s_C(m)$; this leads to the following formula
\begin{equation}\label{hashformula}
\#[\ov s](m) = {m \choose s_0, \ldots, s_{p-1}, s_C(m)} \cdot {n-m \choose l_0-s_0, \ldots, l_{p-1}-s_{p-1}, l_{C} - s_C(m)}.
\end{equation}
 Now, when fixing $s_0, \dots, s_{p-1}$, we will use \cref{period-newton} to show that the function $m \mapsto \#[\ov s](m) \bmod q$ is periodic with period $q^{k_q}$ in the interval $\{0, \dots, n\}$ where  $k_q$ is the smallest integer such that $q^{k_q} > l_0 +\ldots + l_{p-1} = n-|C|$.
 The periodicity is immediate for the first element of the  product as $l_0 +\ldots + l_{p-1} \geq s_0 + \ldots s_{p-1}$. Let $s_i' = l_0 - s_0$ by \cref{period-newton}. One can see that the values of the second element of the product produce, in the interval $[0..n]$,  a reversed sequence compared to the one produced by 

$${m \choose s_0', \ldots, s_{p-1}', m-(s_0' + \ldots + s_{p-1}')}$$

\noindent Indeed, $l_C  - s_C(m) = (n-(l_0 + \ldots + l_{p-1}) - (m - s_0 - \ldots - s_{p-1}) = (n-m) - (s_0' + \ldots s_{p-1}') = $ and $n-m$ plays the role of $m$ in the reversed sequence. As periodicity of any given sequence on the interval $[0..n]$ is preserved under reversing, and as $q^{k_q} > l_0 +\ldots + l_{p-1} \geq s_0' + \ldots + s_{p-1}'$, we get the desired periodicity of $\#[\ov s](m)$ (as the period transfers to the product).
 
Now, one could argue that then the sum 
$$\po s[G;\acpt](m) = \left(\sum_{(\ov s) \in \chi[G; \acpt](m)} \#[\ov s](m)\right) (\mbox{mod } q)$$
\noindent must be periodic, as it is just a sum of elements that are periodic. Unfortunately, $\chi[G; \acpt](m)$ selects elements of the sum depending on $m$ (i.e.\ the $s_i$ depend on $m$). We need to address that. 

 Note that in \eqref{chicondition} we can drop the condition $0 \leq s_C(m) \leq l_C$ because whenever $s_C(m) < 0 $ or $s_C(m) > l_C$ the formula \eqref{hashformula} returns value $0$ anyway (as multinomial coefficient takes value $0$ whenever $s_0 + \ldots + s_{p-1} > m$ or $s_0' + \ldots + s_{p-1}' > m$). So we can effectively get rid of $s_C(m)$ in $\chi$ to get an updated definition
\begin{equation}\label{newchicondition}
\ov s \in \chi'[G; \acpt](m)  \iff
\begin{cases}
\sum_{j=1}^d t_j \cdot {s_C(m) \choose j} + \sum_{j=1}^{q-1} j \cdot s_j\ (\mbox{mod}\ p) = \acpt\\
0 \leq s_i \leq l_i \mbox{ for } i \in \{0, \ldots, p-1\}  \\
\end{cases}
\end{equation}
and maintain the value of the sum, i.e.\ $\left(\sum_{(\ov s) \in \chi'[G; \acpt](m)} \#[\ov s](m)\right) = \left(\sum_{(\ov s) \in \chi[G; \acpt](m)} \#[\ov s](m)\right)$.
Let $K$ be the smallest integer satisfying $p^K > \max(l_0 + l_1 + \ldots + l_{p-1}, d)$. We further modify the definition of $\chi'$ to create $\chi^*$ in the following way
\begin{equation}\label{chiprime}
\ov s \in \chi^*[G; \acpt](m)  \iff
\begin{cases}
\sum_{j=1}^d t_j \cdot {s_C(m) + p^K \choose j} + \sum_{j=1}^{q-1} j \cdot s_j\ (\mbox{mod}\ p) = \acpt\\
0 \leq s_i \leq l_i \mbox{ for } i \in \{0, \ldots, p-1\}  \\
\end{cases}
\end{equation}
We claim that for all $m$
$$ 
\left(\sum_{(\ov s) \in \chi'[G; \acpt](m)} \#[\ov s](m)\right) (\mbox{mod } q) = \left(\sum_{(\ov s) \in \chi*[G; \acpt](m)} \#[\ov s](m)\right) (\mbox{mod } q) 
$$
There are $3$ cases we need to consider. 
\begin{enumerate}
\item When some fixed $\ov s$ belongs to both $\chi'[G; \acpt](m)$ and $\chi^*[G; \acpt](m)$, then $\#[s](m)$ cancels out from both sides of the equation. Similarly if $\ov s$ does not belong to either of the sets, we do not have $\#[s](m)$ on either of sides of the equation.
\item If $\ov s \in \chi'[G; \acpt](m)$ and $\ov s \not \in \chi^*[G;\acpt](m)$  or $\ov s \not \in \chi'[G; \acpt](m)$ and $\ov s \in \chi^*[G;\acpt](m)$, there must be some $j$ such that ${s_C(m) \choose j} \neq {s_C(m)+p^K \choose j}$.
This can only be the case if $s_C(m)$ is negative:
otherwise, because $p^K > d \geq j$, from the periodicity of function $a \mapsto { a \choose j} \bmod p$ (for natural numbers $a\geq 0$), we would get that ${s_C(m) \choose j} = {s_C(m) + p^K \choose j} \bmod p$ and, hence, the conditions for $\chi'[G; \acpt](m)$ and $\chi^*[G; \acpt](m)$ would be identical from the perspective of $\ov s$. But when $s_C(m) < 0$, then $\#[s](m)$ is zero due to definition of multinomial coefficient, so it does not contribute to any of the sides anyway.
\end{enumerate}
So we get that 
$$\po s(G;\acpt)(m) = \left(\sum_{(\ov s) \in \chi^*[G; \acpt](m)} \#[\ov s](m)\right) (\mbox{mod } q)$$

But now the formula \eqref{chiprime} gives ranges $0 \leq s_j \leq l_j$ for $j\in \{0, \dots, p-1\}$; thus, we conclude that $s_C(m) + p^K$ is always positive (since $s_C(m) + p^K = m + p^K - (s_0 + \ldots + s_{p-1}) \geq m \geq 0$). So $m \mapsto {s_C(m) + p^K \choose j} \bmod p$ is periodic with period $p^{k_p} > d$ where $k_p$ is the smallest integer with $p^{k_p} > d$. Looking at the definition of $\chi^*$ we conclude:

\newcommand{\remm}{j}
\begin{fact}
    Let $\ov s \in \N^{p}$. For $m\geq0$ we have 
    \[
     \ov s \in \chi^*[G;\acpt](m) \iff \ov s \in \chi^*[G;\acpt](m + p^{k_p})
    \]
\end{fact}
 \noindent Hence, the condition $\ov s \in \chi^*(G;\acpt)$ depends not really on $m$, but on the remainder of $m$ modulo $p^{k_p}$. So now, for all integers $\remm\in \{0, \ldots, p^{k_p}-1\}$ we define $\chi^{*}_\remm[G;\acpt]$ as $\chi^{*}[G;\acpt](\remm)$ in order to obtain that $\chi^*[G;\acpt](m)= \chi^{*}_{\remm} [G;\acpt]$ for $\remm = m \bmod p^{k_p}$.
Now we can see that $\po s(G;\acpt)(m)$ is periodic with period $p^{k_p}\cdot q^{k_q}$:
\begin{align*}    
\po s(G;\acpt)(m + p^{k_p}\cdot q^{k_q}) &= 
\sum_{(\ov s) \in \chi^*_{\remm}[G; \acpt]} \#[\ov s](m + p^{k_p}\cdot q^{k_q}) \\
&= \sum_{(\ov s) \in \chi^*_{\remm}[G; \acpt]} \#[\ov s](m)  \\
&= \po s(G;\acpt)(m)\ (\mbox{mod } q).
\end{align*}
\noindent The first and the third equality comes from periodicity of the condition $\chi^*$  and the fact that $m$ and $m + p^{k_p}\cdot q^{k_q}$ give the same rest modulo $p^{k_p}$ and the second one comes from the equality of rests modulo $q^{k_q}$ and periodicity of $\#(s)(m)$.
\end{proof}




\end{document}